\documentclass{article}
\usepackage{graphicx}
\usepackage[utf8]{inputenc}
\usepackage{braket}
\usepackage[top=0.75in,left=0.85in,right=0.85in,bottom=0.75in]{geometry}
\usepackage{amsfonts}
\usepackage[T1]{fontenc}
\usepackage{verbatim}
\usepackage{hyperref}
\usepackage{amsthm}
\usepackage{amsmath}
\usepackage{amssymb}
\newtheorem{theorem}{Theorem}

\newcounter{main}[section]

\newtheorem{definition}[main]{Definition}
\newtheorem{example}[main]{Example}
\newtheorem{claim}[main]{Claim}
\newtheorem{proposition}[main]{Proposition}

\newtheorem{lemma}[main]{Lemma}
\newtheorem{corollary}[theorem]{Corollary}

\newtheorem{remark}[main]{Remark}

\usepackage{vwcol}
\usepackage{tikz}
\usepackage{tikz-cd}
\usepackage{mathtools}
\usepackage{tasks}
\usepackage{authblk}
\usepackage{appendix}
\usepackage{algorithm}
\usepackage{algpseudocode}
\usepackage{makecell}
\usepackage{wrapfig}
\usepackage{xspace}
\usetikzlibrary{calc}
\usetikzlibrary{decorations.pathreplacing,calligraphy}

\usepackage[sorting=nyt, 
style=alphabetic, 
backref, 
maxbibnames=20, 
doi=false,
url=true]{biblatex}
\DeclareMathOperator{\supp}{supp}
\DeclareMathOperator{\ev}{ev}
\DeclareMathOperator{\RM}{RM}
\DeclareMathOperator{\punc}{punc}

\DeclareMathOperator{\argmin}{argmin}

\title{Improved Quantum Codes with Transversal $T$ Gates}
\date{\today}
\author[1,2]{Adam Wills\textsuperscript{*}}
\affil[1]{Center for Theoretical Physics — a Leinweber Institute, \protect\\ Massachusetts Institute of Technology, Cambridge, MA\vspace*{5mm}}
\affil[2]{IBM Research, IBM T. J. Watson Research Center, NY}
\begin{document}

\maketitle

\begin{abstract}
    In this work, we study quantum CSS codes with transversal $T$ gates. Here, $T$ gate transversality is meant in the strongest sense; the application of physical $T$ to every physical qubit yields logical $T$ on every logical qubit, without Clifford corrections. Despite the importance of the $T$ gate in fault-tolerant quantum computation, the parameters of asymptotic families of such codes have not been improved since the work of Hastings and Haah in 2017~\cite{sublogarithmic_arxiv}, and Haah in 2018~\cite{haah2018towers}. In this work, we significantly broaden the achievable parameters of quantum code families with transversal $T$ gates, both expanding the regime of achievable polynomial rate and distance, and constructing such codes with constant rate and growing distance; this is the first time the latter has been achieved, even when allowing Clifford corrections after the transversal $T$ gate. These are also the first codes achieving $\gamma \to 0$ for a code with a transversal $T$ gate, where $\gamma$ is the overhead exponent of magic state distillation.

    To do this, we develop a framework of divisible decreasing monomial codes, punctured at a downward-closed set on the Boolean hypercube to create logical qubits. We prove a closed-form expression for the distance of such a code punctured at such a set, which may be of independent interest. We first instantiate this with an explicit construction based on weighted Reed-Muller codes, puncturing at low Hamming-weight points, and then with a randomised construction, where a small random set of points is protected from the puncturing to save quantum code distance, achieving improved parameters.
\end{abstract}
\vfill
\footnoterule
\noindent\textsuperscript{*}\texttt{a\_wills@mit.edu}
\pagebreak
\tableofcontents

\section{Introduction}

The theory of quantum error correction~\cite{shor1995scheme}, and the related field of fault-tolerant quantum computing~\cite{aharonov1997fault}, attempt to design quantum codes, and quantum fault tolerance protocols, with a mind to the eventual construction of a large-scale fault-tolerant quantum computer with low spacetime overhead. To support this effort, theoretical work has sought to understand what quantum codes, and quantum fault tolerance protocols, can and cannot achieve, in various well-motivated settings.

A central notion in these fields is that of transversality. The most general notion of transversality is as follows: on some quantum codes with the right structure, one may act on the physical qubits with a constant-depth circuit, and in doing so execute some useful logical operation on the logical qubits. The operation is inherently fault tolerant, since errors on the physical qubits may only spread by a constant amount in the operation, and is naturally low-overhead in space and time. Unfortunately, the Eastin-Knill theorem~\cite{eastin2009restrictions} says that no quantum code may support a universal transversal gateset. In other words, universal fault-tolerant logic will at some stage require some more expensive operation.

Many natural quantum codes such as the Steane $7$-qubit code~\cite{steane1996error}, and the 2D colour code~\cite{bombin2006topological}, support transversal gates restricted to the Clifford group. It is found to be much harder to construct quantum codes supporting non-Clifford transversal gates. Such codes are well-motivated, since the remaining fault-tolerant operation to achieve universality would be some Clifford operation, for which the corresponding fault-tolerant operation would generally be more lightweight: distilling and teleporting a Hadamard gate, for example.

In this direction, recent works constructed the first asymptotically good families of quantum codes with transversal non-Clifford gates. However, the gates that make up the physical transversal operation, and get implemented at the logical level, were multi-qubit gates such as $CCZ$~\cite{golowich2025asymptotically,nguyen2024good}, or larger gates~\cite{wills2024}. This is not ideal, as such multi-qubit gates are typically more difficult to execute at high fidelity in hardware, and the corresponding logical operation may be less directly useful. Moreover, multi-qubit transversal gates do not imply transversally addressable gates at lower levels in the Clifford hierarchy, for which additional structure is required~\cite{he2025quantum,he2025asymptotically}.\footnote{By contrast, for the single-qubit non-Clifford gates we are going to study, the implication is immediate; see Remark~\ref{rmk:addressability}.} While the codes constructed in the works~\cite{wills2024,golowich2025asymptotically,nguyen2024good} were not LDPC, and so are challenging to implement at the physical level, the lessons learned therein inspired an active, ongoing line of work seeking to construct better quantum LDPC codes with transversal non-Clifford gates~\cite{golowich2025quantum,scruby2024quantum,zhu2025transversal,golowich2025near,zhu2025non,zhu2026non,kobayashi2026clifford,breuckmann2026cups,li2025poincar,li2026transversal}, which could be run at the physical level: a highly desirable prospect.

Like the earlier works~\cite{wills2024,golowich2025asymptotically,nguyen2024good}, we aim to develop the theory around quantum codes with non-Clifford transversal gates. However, we make the crucial restriction that the non-Clifford gate we want our quantum codes to support transversally is the one-qubit gate $T \coloneq \text{diag}(1,e^{i\pi/4})$: the most widely-considered non-Clifford gate in fault-tolerant quantum computing given its simplicity, ease of consumption, how easily it may be executed in many hardware modalities, and the fact that its transversality implies the transversal addressability of gates in lower levels of the Clifford hierarchy. Concretely, we wish to study asymptotic families of quantum stabiliser codes\footnote{We restrict ourselves to stabiliser codes as it is often possible to construct non-stabiliser codes supporting exotic transversal gates~\cite{kubischta2023family,pollatsek2004permutationally,kubischta2024permutation,kubischta2024quantum,zhang2025transversal}, and in particular we suspect it should be possible to construct even further improved quantum codes with transversal $T$ gates by relaxing our restriction. However, stabiliser codes are almost exclusively considered in practice because it is not clear how crucial fault tolerance procedures, such as fault-tolerant syndrome extraction, may be executed on non-stabiliser codes. These are well-developed and well-understood procedures for stabiliser codes~\cite{Gottesman}. All codes constructed in this work will in fact be quantum CSS codes~\cite{calderbank1996good,steane1996multiple}}~\cite{gottesman1997stabilizer} for which the application of the $T$ gate separately on each physical qubit (the physical transversal $T$ gate) executes the logical $T$ gate separately on each logical qubit (the logical transversal $T$ gate). 

While our main motivation is simply theory development, especially informing potential development of quantum LDPC codes with transversal $T$ gates, there are several other possible impacts of this work. First, non-LDPC codes like ours can be run at the logical level of some LDPC code, for example as a magic state distillation scheme~\cite{MSD}. Because we do not allow Clifford corrections after the transversal $T$ gate, our protocols imply particularly fast protocols for distilling $T$ gates to $T$ gates. It will also turn out that the codes we construct with transversal $T$ gates will have the best known parameters, even when allowing Clifford corrections. More generally than magic state distillation, concatenation is emerging as a promising candidate for building quantum architectures~\cite{pattison2025hierarchical,gidney2025yoked,wills2026concatenating}. To perform computation with low overhead on such architectures, high-performance quantum codes with transversal non-Clifford gates may be desirable, where there is no restriction on the codes being LDPC.

There are two further impacts we envision of this work, which are specific to the way we go about constructing our codes. First, we construct our codes using punctured decreasing monomial codes. To establish the distance of the resulting quantum code, we provide a closed-form expression for the distance of a classical Boolean decreasing monomial code after puncturing at a downward-closed set on the Boolean hypercube, see Lemma~\ref{lem:monomial_puncture_distance}, which is new to the best of our knowledge. Given how widely-considered decreasing monomial codes are for classical communication (when considered as polar codes)~\cite{arikan2009channel,bardet2016algebraic,bioglio2020design}, this may be of independent interest. There is a further place, this time in the quantum literature, where the study of punctured multivariate evaluation codes may be of interest. Indeed, a key ingredient of the recent construction in~\cite{li2026transversal} was the proof that tensor products of randomly-punctured Reed-Solomon codes are two-way product expanding. However, the achieved field size is very large. One natural approach for reducing the field size is to consider punctured multivariate evaluation codes such as algebraic geometry or Reed-Muller/monomial codes. The key technical observation of this work is that the distance of punctured multivariate decreasing monomial codes is particularly well behaved when the puncture set is itself a downward-closed set. Thus, it is possible that our work on punctured decreasing monomial codes could influence that line of work.

\subsection{Overview of Results}

The main results of this work are to significantly expand the regime of parameters achievable by asymptotic families of quantum codes supporting transversal $T$ gates, both in the regime of polynomial rate and distance (see Theorem~\ref{thm:sub_constant_explicit} for the parameters we achieve with explicit families and Theorem~\ref{thm:sub_constant_exist} for the parameters we achieve with not necessarily explicit families), and constructing such codes in the regime of constant rate and growing distance for the first time (see Theorem~\ref{thm:constant_explicit} --- these code families are explicit). For concreteness, we state our notion of $T$ gate transversality in a definition here.
\begin{definition}[$T$ gate transversality]\label{def:T_gate_transversality}
    We say that a quantum code supports a transversal $T$ gate if the application of the physical $T$ gate separately on every physical qubit (the physical transversal $T$ gate) executes the logical $T$ gate separately on every logical qubit (the logical transversal $T$ gate).
\end{definition}

\begin{figure}[H]
        \centering
        \includegraphics[width=0.65\linewidth]
        {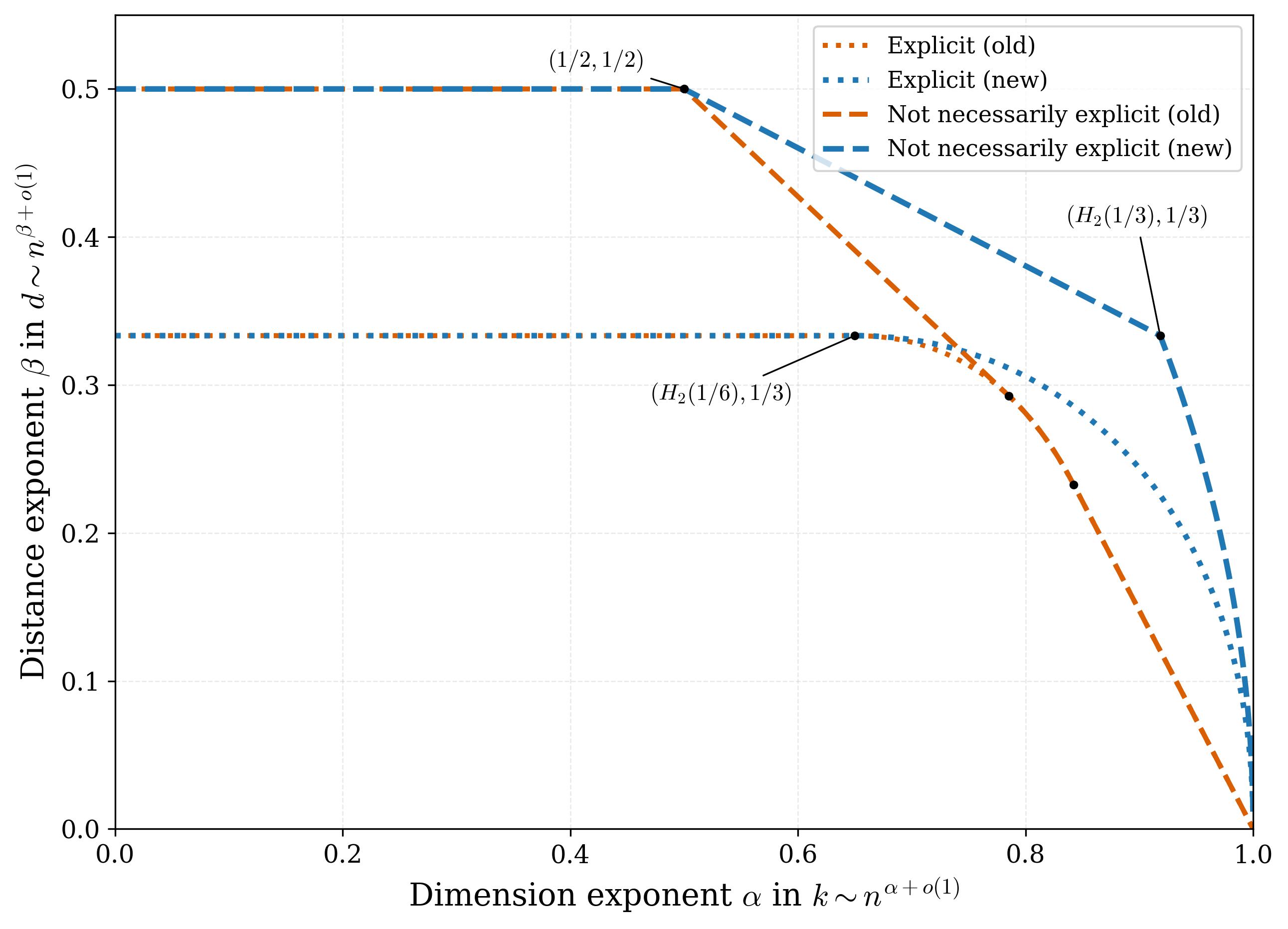}
    \caption{The boundary of the achievable dimension and distance exponents $(\alpha,\beta)$ in the polynomial rate and distance regime $k = \Omega\left(n^{\alpha + o(1)}\right)$, $d = \Omega\left(n^{\beta + o(1)}\right)$, for quantum CSS codes with transversal $T$ gates. We show both explicit constructions, and not necessarily explicit constructions, both new in this work, and previously known. We give the exact formulae for the boundaries of explicitly achievable exponents $\beta(\alpha)$, both new and existing, in Theorem~\ref{thm:sub_constant_explicit} and Remark~\ref{rmk:explicit_old_sub_constant}, and for not necessarily explicit families in Theorem~\ref{thm:sub_constant_exist} and Remark~\ref{rmk:exist_old_sub_constant}, respectively.}
    \label{fig:sub_constant_achievable}
\end{figure}
\begin{figure}[H]
        \centering
        \includegraphics[width=0.65\linewidth]
        {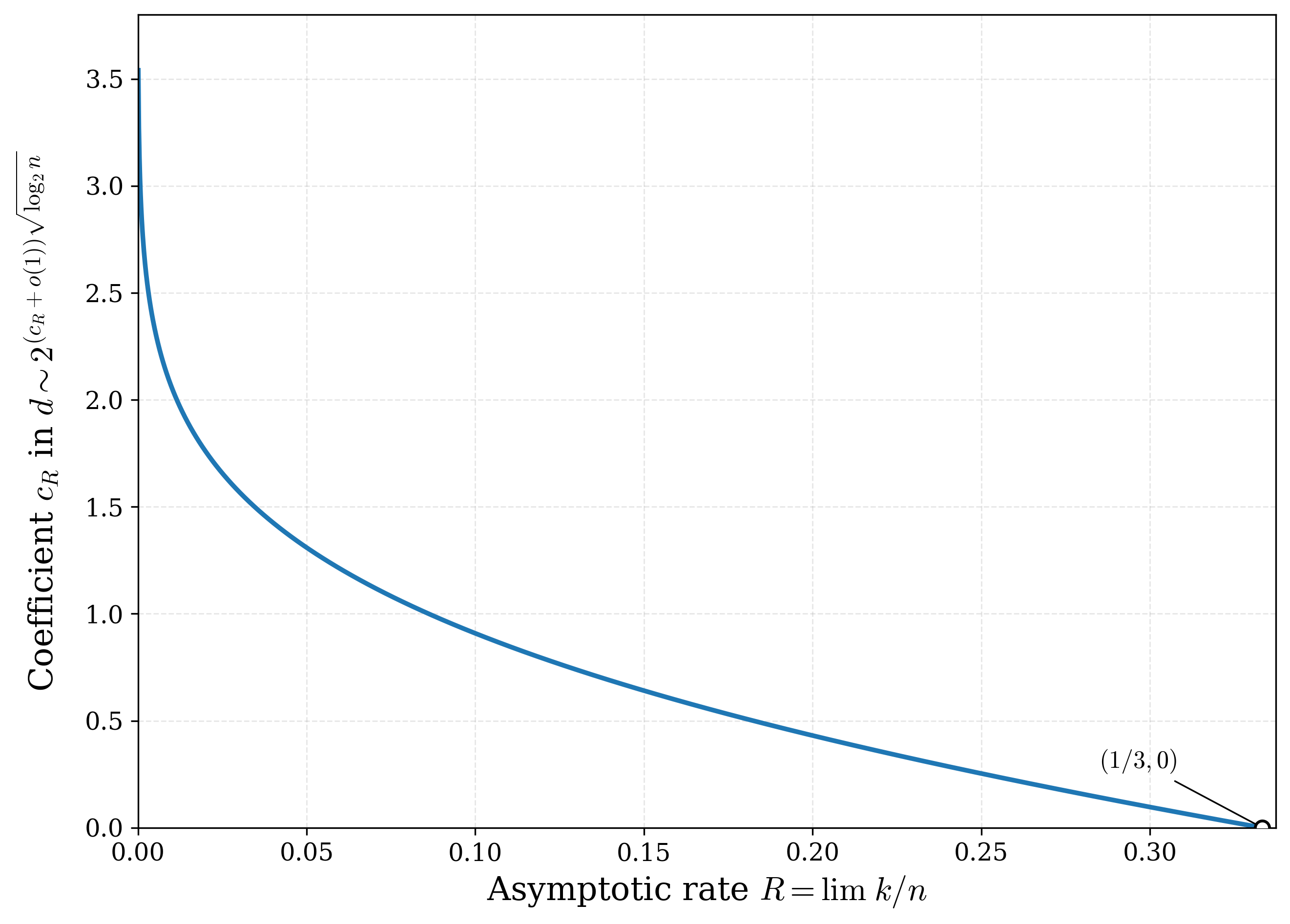}
    \caption{The boundary of the achievable asymptotic parameters in the constant-rate regime. For a given constant rate $R \in (0,1/3)$, we plot $c_R$, where the achievable distance scaling is the growing sub-polynomial expression $d = \Omega\left(2^{\left(c_R+o(1)\right)\sqrt{\log_2n}}\right)$. We do not find a difference between explicit families and not necessarily explicit families in the constant-rate regime, and this constant-rate/growing distance parameter regime was not accessible before this work, hence the single line on the figure. The exact formula for the achievable parameters in this regime is given in Theorem~\ref{thm:constant_explicit}.}
    \label{fig:constant_achievable_parameters}
\end{figure}
\pagebreak
\begin{remark}\label{rmk:T_gate_transversality}
    Note that it is common to consider $T$ gate transversality as the situation where physical $T$ on each physical qubit executes logical $T^\dagger$ on each logical qubit~\cite{hastings2018distillation}. We consider this situation to be the same as Definition~\ref{def:T_gate_transversality}, and do not distinguish between them throughout the paper. The reason is that one may move from one notion to the other in the same quantum CSS code, with a different choice of logical basis. Concretely, one may make the necessary basis change by imagining acting with the $X$ operator on every logical qubit; this makes the necessary change because $XT^\dagger X \propto T$.
\end{remark}

\begin{remark}\label{rmk:addressability}
    With straightforward generalisations of the notions considered in this paper, we may also construct improved quantum codes with transversal $R_\ell \coloneq \text{diag}(1,e^{2\pi i/2^\ell}) = T^{\frac{1}{2^{\ell-3}}}$ gates, for all $\ell \geq 3$. $R_\ell$ is in the $\ell$-th level of the Clifford hierarchy. Moreover, since $R_\ell$ is a single-qubit gate, and $R_\ell X R_\ell^\dagger X = R_{\ell-1}$ (up to a global phase), a code with a transversal $R_\ell$ gate implies a code with a transversally addressable $R_{\ell-1}$ gate. An analogous implication does not hold for transversal, and transversally addressable multi-qubit gates, for which additional structure is required~\cite{he2025quantum,he2025asymptotically}.
\end{remark}
\begin{remark}\label{rmk:explicitness}
    An explicit family of quantum codes is one for which there is a polynomial time algorithm constructing the code (more concretely, a generator matrix for the code). Such constructions are usually arrived at via constructive means. By contrast, one may show the existence of not necessarily explicit families via some randomised argument; see Section~\ref{sec:random_construction}.
\end{remark}

Turning to our main results, we begin by stating our result on the parameters we may achieve in the polynomial rate and distance regime with explicit families. The new and prior boundary of parameters achievable by explicit families is shown in Figure~\ref{fig:sub_constant_achievable}.
\begin{theorem}\label{thm:sub_constant_explicit}
    Consider\begin{equation}\label{eq:explicit_exponent_boundary}
    \beta_{\text{explicit}}(\alpha) = \begin{cases}
        \frac{1}{3} &\text{ if } 0 < \alpha \leq H_2(1/6)\\
        \frac{1-2q^*(\alpha)}{3} &\text{ if } H_2(1/6) < \alpha < 1
    \end{cases},
\end{equation}
where
\begin{equation}
    q^*(\alpha)\text{ is the unique positive solution to }1-2q = (1+q)H_2\left(\frac{3H_2^{-1}(\alpha)}{1+q}\right),
\end{equation}
    where $H_2$ is the binary entropy function, and $H_2^{-1}: (0,1) \to (0,1/2)$ is its inverse. For any constant $\alpha \in (0,1)$, there exists an explicit family of quantum CSS codes with transversal $T$ gates with parameters 
    \begin{equation}
        [[n,\Omega(n^{\alpha+o(1)}), \Omega(n^{\beta_{\text{explicit}}(\alpha)+ o(1)})]].
    \end{equation}
\end{theorem}
\begin{remark}\label{rmk:explicit_old_sub_constant}
    Before this work, the boundary of the achievable exponents for explicit codes was given by
    \begin{equation}
        \beta_{\text{explicit, old}}(\alpha) = \begin{cases}
            \frac{1}{3} &\text{ if } 0<\alpha \leq H_2(1/6)\\
            \frac{1}{3}H_2(3H_2^{-1}(\alpha)) &\text{ if }H_2(1/6) < \alpha \leq \alpha^*\\
            \frac{1-\alpha}{\gamma^*} &\text{ if } \alpha^* \leq \alpha < 1,
        \end{cases}
    \end{equation}
    instead of Equation~\eqref{eq:explicit_exponent_boundary}, where $\alpha^* = H_2(p^*) \approx 0.842$, $p^* = \argmin_{p \in [1/6,1/3)}\frac{3(1-H_2(p))}{H_2(3p)} \approx 0.271$ and $\gamma^* = \min_{p \in [1/6,1/3)}\frac{3(1-H_2(p))}{H_2(3p)} \approx 0.678$. This function is obtained by considering Hastings and Haah's codes~\cite{hastings2018distillation}, and taking polynomially many copies thereof. Note that Haah's codes~\cite{haah2018towers} are not included here since they are not explicit.
\end{remark}

Theorem~\ref{thm:sub_constant_explicit} will be proved in Subsection~\ref{subsec:explicit_asymptotic_sub_constant_rate} using the preceding construction and results. We next state our second main result on the parameters we may achieve in the constant rate, growing distance regime with explicit families. The achievable parameters in this regime are plotted in Figure~\ref{fig:constant_achievable_parameters}.
\begin{theorem}\label{thm:constant_explicit} 
    For every constant $R \in (0,1/3)$, there exists an explicit family of quantum CSS codes with transversal $T$ gates and parameters
    \begin{equation}
        \left[\left[n,\left(R + o(1)\right)n,\Omega\left(2^{\left(c_R + o(1)\right)\sqrt{\log_2(n)}}\right)\right]\right],
    \end{equation}
    $c_R = -\Phi^{-1}\left(\frac{2R}{1+R}\right) > 0$, where $\Phi$ is the cumulative distribution function of the standard Gaussian.
\end{theorem}
Theorem~\ref{thm:constant_explicit} will be proved in Subsection~\ref{subsec:explicit_constant_rate} using the preceding construction and results. It is the first result to establish that the overhead exponent of magic state distillation~\cite{bravyi2012magic} $\gamma = \frac{\log(n/k)}{\log(d)} \to 0$ using a quantum code with a transversal $T$ gate (or even a transversal $T$ gate that allows Clifford corrections),\footnote{This includes results where the physical transversal $T$ gate induces a logical $T$ gate on some subset of the logical qubits~\cite{reddy2026realizing}.} improving upon the value of $\gamma \approx 0.678$ in Hastings and Haah's paper~\cite{hastings2018distillation}. Previously, $\gamma\to 0$ had only been achieved with codes supporting larger multi-qubit or qudit gates~\cite{krishna2018towards,wills2024,golowich2025asymptotically,nguyen2024good,nguyen2025quantum}.\footnote{Since it is possible to make those codes asymptotically good, for those larger gates it is possible to obtain $\gamma = 0$~\cite{wills2024,golowich2025asymptotically,nguyen2024good}.}
\begin{corollary}
    There exists an explicit family of quantum codes supporting a transversal $T$ gate whose magic state distillation exponent $\gamma = \frac{\log(n/k)}{\log(d)}$ tends to zero.
\end{corollary}

By relaxing the requirement for the family to be explicit, we find that we obtain further new parameters in the polynomial rate and distance regime. The new and prior boundary of achievable parameters is shown in Figure~\ref{fig:sub_constant_achievable}.
\begin{theorem}\label{thm:sub_constant_exist}
    Consider
    \begin{equation}\label{eq:exist_exponent_boundary}
        \beta_{\text{exist}}(\alpha) = \begin{cases}
            1/2 &\text{ if } 0 < \alpha \leq 1/2\\
            1/2 - \frac{\alpha-1/2}{6\left(H_2(1/3)-1/2\right)} &\text{ if } 1/2 \leq \alpha \leq H_2(1/3)\\
            1-2H_2^{-1}(\alpha) &\text{ if } H_2(1/3) \leq \alpha < 1
        \end{cases},
    \end{equation}
    where $H_2$ is the binary entropy function and $H_2^{-1}:(0,1)\to(0,1/2)$ is its inverse. For any constant $\alpha \in (0,1)$, there exists a family of quantum CSS codes with transversal $T$ gates with parameters 
    \begin{equation}
        [[n,\Omega(n^{\alpha+o(1)}), \Omega(n^{\beta_{\text{exist}}(\alpha)+ o(1)})]].
    \end{equation}
\end{theorem}
\begin{remark}\label{rmk:exist_old_sub_constant}
    Before this work, the boundary of the achievable exponents for not necessarily explicit codes was given by
    \begin{equation}
        \beta_{\text{exist, old}}(\alpha) = \begin{cases}
            \frac{1}{2} &\text{ if } 0<\alpha \leq \frac{1}{2}\\
            \frac{1}{2} + s^\dagger\left(\alpha-\frac{1}{2}\right) &\text{ if }\frac{1}{2}<\alpha \leq \alpha^\dagger\\
            \frac{1}{3}H_2(3H_2^{-1}(\alpha)) &\text{ if }\alpha^\dagger < \alpha \leq \alpha^*\\
            \frac{1-\alpha}{\gamma^*} &\text{ if } \alpha^* \leq \alpha < 1,
        \end{cases}
    \end{equation}
    instead of Equation~\eqref{eq:exist_exponent_boundary}, where $\alpha^\dagger = H_2(p^\dagger) \approx 0.785$, $p^\dagger \approx 0.234$ is the unique solution in $(1/6,1/3)$ to
    \begin{equation}
        \frac{\frac{1}{3}H_2(3p)-\frac{1}{2}}{H_2(p)-\frac{1}{2}}
        =
        \frac{H_2'(3p)}{H_2'(p)},
    \end{equation}
    and
    \begin{equation}
        s^\dagger
        =
        \frac{\frac{1}{3}H_2(3p^\dagger)-\frac{1}{2}}{H_2(p^\dagger)-\frac{1}{2}}
        \approx -0.727.
    \end{equation}
    Moreover, $\alpha^* = H_2(p^*) \approx 0.842$, $p^* = \argmin_{p \in [1/6,1/3)}\frac{3(1-H_2(p))}{H_2(3p)} \approx 0.271$ and $\gamma^* = \min_{p \in [1/6,1/3)}\frac{3(1-H_2(p))}{H_2(3p)} \approx 0.678$. This function is obtained by considering Hastings and Haah's codes~\cite{hastings2018distillation}, Haah's divisible tower codes~\cite{haah2018towers}, concatenating the two constructions, and taking polynomially many copies thereof.
\end{remark}
\begin{remark}
    We found that our randomised construction (see Section~\ref{sec:random_construction}) did not improve on the parameters of our explicit construction in the constant-rate regime (see Theorem~\ref{thm:constant_explicit}).
\end{remark}

\subsection{Overview of Techniques}
We will now overview the techniques we use to construct quantum codes with transversal $T$ gates. 

\subsubsection{\texorpdfstring{$T$}{} Gate Transversality from Classical Divisibility}\label{subsubsec:T_gate_classical_divis}

We will begin with an overview of the so-called ``$X$-generator matrix'' formalism in which we work, which is standard in the literature~\cite{bravyi2012magic,hastings2018distillation}, and show how it may be instantiated with classical divisible codes to obtain quantum codes with transversal $T$ gates.

One considers a binary matrix $G \in \mathbb{F}_2^{s \times n}$ with $s \leq n$, and some integer $k \in \{1, \ldots, s-1\}$; usually one would make the restriction that $G$ is full-rank. The first $k$ rows of $G$ are called the \textit{logical rows} of $G$, and the latter $s-k$ rows of $G$ are called the \textit{stabiliser rows} of $G$. We use $G_1$ and $G_0$ to denote the sub-matrices formed by the logical rows and stabiliser rows, respectively, and $\mathcal{G}_1, \mathcal{G}_0, \mathcal{G}$ to denote the binary vector spaces spanned by the logical rows, stabiliser rows, and all rows of $G$, respectively. From $G$, we define a quantum CSS code~\cite{calderbank1996good,steane1996multiple} $\mathcal{Q}[G]$ of length $n$, whose $X$ stabilisers are given by $\mathcal{G}_0$, and whose $Z$ stabilisers are given by $\mathcal{G}^\perp$. By placing different requirements on $G$, one is able to obtain that $\mathcal{Q}[G]$ has different, desirable transversality properties. We denote the rows of $G$ by $(g^a)_{a=1}^s$, by $|v|$ the usual Hamming weight of a binary vector $v$, and by $v_1*v_2$ the overlap of two binary vectors $v_1$ and $v_2$, that is, their componentwise multiplication. We consider the following properties. For all $a,b,c \in [s]$,
\begin{align}
    |g^a| &\equiv \begin{cases}
        7 \pmod 8 &\text{ if } 1 \leq a \leq k\label{eq:G_property_first}\\
        0 \pmod 8 &\text{ if } k+1 \leq a \leq s
    \end{cases},\\
    |g^a * g^b| &\equiv 0 \pmod 4 \text{ for all } 1 \leq a < b \leq s,\\
    |g^a * g^b * g^c| &\equiv 0 \pmod 2 \text{ for all } 1 \leq a < b < c \leq s.\label{eq:G_property_last}
\end{align}
In words, we want the overlap of any three rows to have even Hamming weight, the overlap of any two rows to have Hamming weight divisible by $4$, and any row to have Hamming weight divisible by $8$, unless it is a logical row, in which case it must have Hamming weight equivalent to $7 \pmod 8$. Given these conditions, one can show~\cite{hastings2018distillation,bravyi2012magic} that $\mathcal{Q}[G]$ has $k$ logical qubits, and that (under a natural choice of logical basis) the physical transversal $T$ gate $T^{\otimes n}$ implements the logical transversal $T^\dagger$, that is, $\overline{T^{\dagger \otimes k}}$. As stated in Remark~\ref{rmk:T_gate_transversality}, we do not distinguish between the situation where $T^{\otimes n}$ implements $\overline{T^{\dagger \otimes k}}$ and $\overline{T^{\otimes k}}$; it is actually easier to consider the former in this exposition. Note that the construction of $\mathcal{Q}[G]$ from the matrix $G$, and the conditions being placed on $G$, turn out to be necessary and sufficient for a quantum CSS code to have a transversal $T$ gate~\cite{rengaswamy2020optimality,hu2022designing}.\footnote{Strictly speaking, this is the necessary and sufficient condition for a ``positively-signed CSS code''~\cite{rengaswamy2020optimality,hu2022designing}, but very similar conditions hold for other signs.}
\begin{remark}
    Weaker, but more easily attained, constraints can be placed on $G$, to obtain weaker properties on $\mathcal{Q}[G]$. For example, the triorthogonality condition~\cite{bravyi2012magic} is a strict weakening, namely, for all $a,b,c \in [s]$,
    \begin{equation}
        |g^a*g^b*g^c| \equiv \begin{cases}
            1 \pmod 2 &\text{ if } 1 \leq a = b = c \leq k\\
            0 \pmod 2 &\text{ otherwise}
        \end{cases}.
    \end{equation}
    This leads to $\mathcal{Q}[G]$ again having $k$ logical qubits, and having a transversal $T$ gate up to Clifford corrections, meaning that there is some (generally non-transversal) Clifford operator on $n$ qubits, $U$, for which $UT^{\otimes n}$ implements $\overline{T^{\otimes k}}$. Interestingly, while the triorthogonality condition is mathematically much weaker than the full requirements above for a transversal $T$ gate, the best asymptotic families of quantum codes satisfying the triorthogonality condition also satisfy the stronger condition for a transversal $T$ gate, both prior to this work~\cite{hastings2018distillation,haah2018towers}, and in this work. In particular, asymptotically good triorthogonal codes for qubits are not known.
    
    One may weaken the condition on $G$ further. Indeed, given two permutations $\pi_1, \pi_2$ on $n$ elements, suppose that for all $a,b,c \in [s]$,
    \begin{equation}
        |g^a*\pi_1(g^b)*\pi_2(g^c)| \equiv \begin{cases}
            1 \pmod 2 &\text{ if } 1 \leq a = b = c \leq k\\
            0 \pmod 2 &\text{ otherwise}
        \end{cases}.
    \end{equation}
    If $G$ satisfies this condition, then $\mathcal{Q}[G]$ supports a transversal $CCZ$ gate between three codeblocks. This condition is weaker still, and in this case that is reflected in the parameters known to be achievable by quantum codes;~\cite{golowich2025asymptotically,nguyen2024good} constructed asymptotically good quantum codes satisfying this condition.
    
    Before moving on, we comment that when $G$ is a matrix over a larger field $\mathbb{F}_q$, and the quantum code $\mathcal{Q}[G]$ is for Galois qudits~\cite{wills2026review}, there is a much richer array of conditions that $G$ can satisfy to achieve different transversality properties on $\mathcal{Q}[G]$~\cite{wills2024,Gong2026BinaryExtensionFields}.
\end{remark}
The most common way, and the way we pursue in this paper, to obtain matrices $G$ with the conditions in Equations~\eqref{eq:G_property_first} to~\eqref{eq:G_property_last}, is to consider a binary classical linear code $C$, for which $\dim C = s$, and write down a generator matrix $\hat{G} \in \mathbb{F}_2^{s \times N}$ for it.  We then choose a set of its coordinates $\Gamma \subseteq [N]$ of size $k = |\Gamma|$. We make such a choice that $C^\perp$ does not have any non-zero codewords supported entirely inside $\Gamma$. In other words, the columns of $\hat{G}$ corresponding to $\Gamma$ are linearly independent, and after row operations and column permutations, $\hat{G}$ may be rewritten as
\begin{equation}\label{eq:semi_systematic_form_overview}
    \hat{G} = \begin{pmatrix}
        I_k & G_1\\
        0 & G_0
    \end{pmatrix},
\end{equation}
where $I_k$ is the $k \times k$ identity matrix, and we obtain the matrix $G = \left(\begin{smallmatrix}G_1\\G_0\end{smallmatrix}\right)$, with $k$ logical rows, $s-k$ stabiliser rows, and $n \coloneq N-k$ columns.

Usually, we require that $C$ has some algebraic properties, which will translate into the desired properties of $G$, and the transversality properties of $\mathcal{Q}[G]$. For our case, suppose that $C$ is an $8$-divisible code, meaning that all of its codewords have Hamming weight divisible by $8$. Then, by the Ward identities~\cite{ward1990weight}, any overlap of two of its codewords have weight divisible by $4$, and any overlap of three of its codewords have weight divisible by $2$, and it follows that the resulting $G$ satisfies Equations~\eqref{eq:G_property_first} to~\eqref{eq:G_property_last}.

While we know that the quantum code has $k$ logical qubits, it in fact turns out that its distance is equal to its $Z$-distance; see~\cite{bravyi2012magic} or the preliminary material in Subsection~\ref{subsec:prelims_X_gen}. The distance of the quantum code $\mathcal{Q}[G]$ is thus
\begin{equation}
    d = d_Z = \min_{v \in \mathcal{G}_0^\perp \setminus \mathcal{G}^\perp}|v|.
\end{equation}
Often, we use the lower bound
\begin{equation}\label{eq:lower_bound_degen}
    d = d_Z \geq d\left(\mathcal{G}_0^\perp\right),
\end{equation}
where $d\left(\mathcal{G}_0^\perp\right)$ is the classical distance of the classical code $\mathcal{G}_0^\perp$. In this paper, we will take this lower bound, and it will turn out to be tight; see Proposition~\ref{prop:downsets_to_quantum_code} and its proof. Notice that, because $\mathcal{G}_0$ is nothing more than the code $C$ shortened at the coordinates $\Gamma$, the code $\mathcal{G}_0^\perp$ is nothing more than the code $C^\perp$ punctured at the coordinates $\Gamma$, that is, $\mathcal{G}_0^\perp = \text{punc}_\Gamma(C^\perp)$.

Since we require $C$ to be an $8$-divisible binary linear code, it is natural to look to binary Reed-Muller codes for $C$; for example, this is how one may construct the $[[15,1,3]]$ code~\cite{knill1996threshold,steane1999quantum,MSD}. The Reed-Muller code $\RM(m,r)$ is $8$-divisible if $3r < m$. It is instructive to consider why.

The codewords of the code $\RM(m,r)$ are in one-to-one correspondence with $m$-variate polynomials over $\mathbb{F}_2$ of degree at most $r$; each variable satisfies $x_i^2 = x_i$ in these polynomials, meaning that each term in each polynomial is some product of a distinct collection of $x_i$.\footnote{Formally, the codewords are in one-to-one correspondence with elements of the polynomial ring $\mathbb{F}_2[x_1, x_2, \ldots, x_m]/(x_1^2-x_1, x_2^2-x_2, \ldots, x_m^2-x_m)$ of degree at most $r$.} To pass from a polynomial to a codeword, the polynomial is evaluated over the entire Boolean hypercube $\mathbb{F}_2^m$ (in particular, the code has length $2^m$), but it is helpful to think of the codewords themselves directly as polynomials. The set of monomials of degree at most $r$, $\left\{x^A: A \subseteq [m]: |A| \leq r\right\}$ (where we use the shorthand $x^A \coloneq \prod_{i \in A}x_i$), forms a basis for $\RM(m,r)$ under this correspondence. It is clear that the Hamming weight of a codeword corresponding to $\prod_{i \in S}x_i$ is $2^{m-|S|}$, for every $S \subseteq [m]$.

Now, the Ward identities~\cite{ward1990weight} tell us that a (binary, linear) classical code $C$ is $8$-divisible if and only if it has a basis, all of whose elements have Hamming weight divisible by $8$, for which any pair have overlap with Hamming weight divisible by $4$, and for which any triple have Hamming weight divisible by $2$; we prove this in Lemma~\ref{lemma:ward_criteria} for completeness. If $f$ and $g$ are polynomials, and $\ev(f)$ and $\ev(g)$ are the corresponding codewords (formed by evaluating them), the crucial fact is that $\ev(f)*\ev(g) = \ev(fg)$, that is, the overlap of the two codewords is the evaluation of the product of the polynomials. As a concrete example, given $A, B \subseteq [m]$, we have $\ev(x^A)*\ev(x^B) = \ev(x^{A \cup B})$. Given the above facts, one can see that $\RM(m,r)$ is $8$-divisible for $3r < m$; in particular, given any $A,B,C \subseteq [m]$ for which $|A|, |B|, |C| \leq r$, we have $|A \cup B \cup C| \leq m-1$, and $|\ev(x^A)*\ev(x^B)*\ev(x^C)|$ is even.

Before the present work, the best-known quantum codes with transversal $T$ gates were given by both the work of Hastings and Haah~\cite{hastings2018distillation} and Haah~\cite{haah2018towers}, which were optimal in different parameter regimes. Our main inspiration in this paper comes from the work of Hastings and Haah~\cite{hastings2018distillation}. Therein, the authors consider the classical code $\RM(m,r)$, for $3r < m$, and consider applying the construction above. Now, the naive thing to do would be to take $C = \RM(m,r)$ for $r = \left\lfloor\frac{m-1}{3}\right\rfloor$, for which $C^\perp = \RM(m,m-r-1)$, and consider any set of its coordinates $\Gamma \subseteq [2^m]$ of size $k = |\Gamma| < 2^{r+1}$. Since $C^\perp$ has distance $2^{r+1}$, the coordinates $\Gamma$ cannot entirely contain a non-zero codeword of $C^\perp$, and so the generator matrix $\hat{G}$ for $\RM(m,r)$ may be written into the form of Equation~\eqref{eq:semi_systematic_form_overview}. One then forms the resulting $G$, and $\mathcal{Q}[G]$. We may then form a lower bound for the quantum code
\begin{equation}\label{eq:loose_distance_bound}
    d \geq d\left(\mathcal{G}_0^\perp\right) = d\left(\text{punc}_\Gamma\left(C^\perp\right)\right) \geq d\left(C^\perp\right) - k = 2^{r+1} - k.
\end{equation}
Of the two lower bounds in Equation~\eqref{eq:loose_distance_bound}, the first was mentioned above in Equation~\eqref{eq:lower_bound_degen}, and the second is a standard fact that puncturing a code in $k$ positions cannot drop its distance by more than $k$. For the construction in~\cite{hastings2018distillation}, the first bound is tight (and the corresponding bound in our work will be tight), but the second bound is not, at least when one makes a good choice of $\Gamma$.\footnote{It is interesting to note that, for the construction of asymptotically good quantum codes with non-Clifford transversal gates such as $CCZ$, and larger gates~\cite{wills2024,golowich2025asymptotically,nguyen2024good}, the actual choice of $\Gamma$ does not matter, beyond its size. In this work, we must ensure that $\Gamma$ has the right structure, not just size.} The breakthrough made in~\cite{hastings2018distillation} is based on making a careful choice for $\Gamma$, and accordingly getting a better value of $d$. Indeed, Hastings and Haah~\cite{hastings2018distillation} make the particular choice $\Gamma = \{A \subseteq [m]: |A| \leq w\}$ for some integer $0 \leq w < r$,\footnote{In words, they puncture the points on the Boolean hypercube whose Hamming weights are at most $w$.} and are able to show that the resulting quantum code has distance $\sum_{i=w+1}^{r+1}\left(\begin{smallmatrix}r+1\\i\end{smallmatrix}\right)$. This can strictly improve the general lower bound in Equation~\eqref{eq:loose_distance_bound}, which is $2^{r+1} - \sum_{i=0}^w\left(\begin{smallmatrix}m\\i\end{smallmatrix}\right)$ for this $\Gamma$.

\subsubsection{Decreasing Monomial Codes, Downset Punctures, and Protected Points}\label{subsubsec:discussion_decreasing_monomial_downset}

To improve upon the construction in Hastings and Haah, we first revisit the Reed-Muller code $\RM(m,r)$. While it is $8$-divisible for $3r<m$, we note that it has a relatively low dimension for an $8$-divisible code. Our simple starting intuition is that we would want to obtain an $8$-divisible code $C$ with as large a dimension as possible, because we cannot puncture more than $\dim(C)$ points ($k < s$ above), and since the number of logical qubits will equal the size of the punctured set $\Gamma$, we should start by obtaining as large an $8$-divisible code $C$ as possible. To do this, we turn to binary monomial codes: natural generalisations of Reed-Muller codes. We give an intuitive overview here; the formal overview is given in the preliminaries. With $2^{[m]}$ the set of subsets of $[m]$, a binary monomial code is specified by a set $\Delta \subseteq 2^{[m]}$. The code, denoted $C(\Delta)$, is the code spanned by (the evaluations of) the monomials $x^A$, over $A \in \Delta$; these monomials form a basis for the code. 

We will restrict ourselves to the well-behaved and widely-studied sub-class of \textit{decreasing} monomial codes, for which $\Delta$ is a downward-closed set, meaning that if $A \in \Delta$, and $B \subseteq A$, then $B \in \Delta$. For brevity, we refer to such a subset of $2^{[m]}$ as a \textit{downset}~\cite{srinivasan2019decoding}. As one example of decreasing monomial codes being nicely behaved, their dual codes are easy to characterise, and are also decreasing monomial codes; one can convince oneself that $C(\Delta)^\perp = C(\Delta^\perp)$, where
\begin{align}
    \Delta^\perp &= \{A \subseteq [m]: A \cup B \neq [m] \text{ for all } B \in \Delta\}\\
    &= \{A \subseteq [m]: \bar{A} \notin \Delta\},\label{eq:dual_downset_expos}
\end{align}
where $\bar{A} = [m]\setminus A$ in the latter expression. We call $\Delta^\perp$ the downset dual to $\Delta$.

Clearly, $\RM(m,r)$ is an example of a decreasing monomial code. However, it is possible to construct very large, in fact constant rate, decreasing monomial codes with the $8$-divisibility property. Indeed, let $\Delta = 2^{[m-3]}$, that is, $\Delta$ is formed of all subsets of $\{1, 2, \ldots, m-3\}$. By the same considerations that showed that $\RM(m,r)$ was an $8$-divisible code for $3r<m$, one can check that $C(\Delta)$ is $8$-divisible with this choice of $\Delta$, and moreover $C(\Delta)$ is a classical code of length $2^m$ and dimension $2^{m-3}$, and therefore has constant rate.

Such an extreme choice as $\Delta = 2^{[m-3]}$ would, however, lead to bad quantum code parameters. The reason is that the code dual to $C(\Delta)$ contains the codeword corresponding to $x_1x_2\ldots x_{m-1}$, which has weight two. Puncturing cannot increase this distance,\footnote{In this work, we do not consider situations where entire codewords (of the dual code) can be punctured away, that is, we only consider situations where puncturing is an injective operation; this means that puncturing cannot increase distance. We discuss this more below, and call the condition that no entire codeword is punctured (that puncturing is injective) the non-annihilation condition.} and so the quantum code has distance at most two, because the quantum code has distance equal to that of the punctured form of $C(\Delta)^\perp$; see Proposition~\ref{prop:downsets_to_quantum_code}. A different, less extreme choice than $\Delta = 2^{[m-3]}$, namely \textit{weighted Reed-Muller codes}~\cite{sorensen1992weighted,geil2013weighted}, will work out well, however. For these codes, one picks a vector of weights $\boldsymbol{s}$, a vector of positive integers of length $m$, and given an integer ``total weight'' $r$, one forms a decreasing monomial code from $\Delta = \{A \subseteq [m]: \sum_{i \in A}s_i \leq r\}$. We will return to this below, where we make the simple choice $\boldsymbol{s} = (\underbrace{1, 1, \ldots, 1}_{m-1}, h)$, for some positive integer $h$.

Motivated by the desire to use decreasing monomial codes for this problem, we need to understand their distance, and their minimum weight codewords. This is well understood (see Theorem 3.9 of~\cite{camps2020polar}), and the distance can be seen as a special case of the much more general footprint bound~\cite{hoholdt1998or,geil2000footprints}). Indeed, the distance of $C(\Delta)$ is simply
\begin{equation}
    d(C(\Delta)) = \min_{B \in \Delta}2^{m-|B|}:
\end{equation}
a direct generalisation of the distance of Reed-Muller codes. Moreover, one can write down examples of minimum-weight codewords. For example, letting $A = \text{argmax}_{B \in \Delta}|B|$ (picking arbitrarily if there are multiple $B$ giving the maximum), the polynomial
\begin{equation}
    \prod_{i \in A}(1+x_i) = \sum_{B \subseteq A}x^B
\end{equation}
yields a codeword of $C(\Delta)$ of weight $2^{m-|A|}$.

At this point, the key observation we will make will be that the distance of punctured decreasing monomial codes is particularly well-behaved, \textit{when the puncture set $\Gamma$ is itself a downset}. In other words, we want to specify two downsets $\Delta, \Gamma \subseteq 2^{[m]}$, and consider $\text{punc}_\Gamma\left(C(\Delta)\right)$. We show a closed-form expression for the distance of this code in Lemma~\ref{lem:monomial_puncture_distance}, under a natural (also closed form) restriction that says that no codewords of $C(\Delta)$ are completely contained inside the set of punctured points $\Gamma$, so that puncturing is an injective code operation; let us refer to this as the non-annihilation condition. Under the non-annihilation condition, we will prove that
\begin{equation}
    d\left(\text{punc}_\Gamma\left(C(\Delta)\right)\right) = \min_{A \in \Delta}\left|\{V \subseteq \bar{A}: V \in \bar{\Gamma}\}\right|.\footnote{The non-annihilation condition in fact exactly says that this right-hand side is positive, which is sensible since we are claiming it as the distance of a code.}
\end{equation}
We can provide intuition here by showing that the distance is at most the right-hand side. Indeed, suppose that $A\in \Delta$ is a choice saturating the minimum on the right-hand side. We consider the polynomial as above, $\prod_{i \in A}(1+x_i)$. Before puncturing, its support is exactly $2^{\bar{A}}$, that is, the set of subsets of $\bar{A}$. After puncturing the points $\Gamma$, its remaining support is $\{V \subseteq \bar{A}: V \in \bar{\Gamma}\}$.

With this exposition, we may finally present the explicit construction of quantum codes with transversal $T$ gates developed in Section~\ref{sec:explicit_construction}. We let $\Delta$ be a downset corresponding to a weighted Reed-Muller code with weights $\boldsymbol{s} = (\underbrace{1, 1, \ldots, 1}_{m-1}, h)$ and some total weight $r$. By taking $h$ larger, we may construct increasingly high-rate decreasing monomial codes with the $8$-divisibility property.\footnote{Note that we in fact found that taking $h$ larger than $1$ only helped for relatively high quantum code rates; for relatively low rates, we do not improve on the parameters of~\cite{hastings2018distillation,haah2018towers}; see Figure~\ref{fig:sub_constant_achievable}.} The distance of our quantum code is at least\footnote{This will in fact turn out to be an equality in our constructions, in other words, the lower bound of Equation~\eqref{eq:lower_bound_degen} will always be tight for us.}
\begin{align}
    d\left(\text{punc}_\Gamma(C(\Delta)^\perp)\right) &= \min_{A \in \Delta^\perp}\left|\{V \subseteq \bar{A}: V \in \bar{\Gamma}\}\right|\\
    &=\min_{A \notin \Delta}\left|\{V \subseteq A: V \in \bar{\Gamma}\}\right|,\label{eq:quantum_distance_before_rewrite}
\end{align}
where going into the second line we use the expression for the dual downset in Equation~\eqref{eq:dual_downset_expos}. For this expression to be valid, we need the non-annihilation condition on $C(\Delta^\perp)$, but this turns out to be easily satisfied by enforcing $\Gamma \subseteq \Delta$. We will go on to show that with the simple choice of
\begin{equation}\label{eq:starting_gamma_form}
    \Gamma = \{A \subseteq [m-1]: |A| \leq w\},
\end{equation}
for some $0 \leq w < r$, which ensures $\Gamma \subseteq \Delta$, we can obtain improved explicit quantum codes with transversal $T$ gates, including those with constant rate and growing distance.

In Section~\ref{sec:random_construction}, we will improve on the parameters in our explicit construction (at least in the regime of polynomial rate and distance), by making the same choice for $\Delta$, but a slightly more fine-grained choice for $\Gamma$. We will still start from a $\Gamma$ of the form of Equation~\eqref{eq:starting_gamma_form}, for some $w$, but we will ``protect'' some points from this puncture set. The intuition for this is as follows. One may rewrite the distance of the quantum code from Equation~\eqref{eq:quantum_distance_before_rewrite} succintly as
\begin{equation}\label{eq:succinct_distance_expression}
    d = \min_{A \notin \Delta}\left|2^A \setminus \Gamma\right|,
\end{equation}
where $2^A$ is the set of subsets of $A \subseteq [m]$. We consider a hypergraph $\mathcal{H}$ on the first $m-1$ variables, where each hyperedge has size $t$, for some integer $t \leq w$. For some integer $y > t$, we will say that $\mathcal{H}$ \textit{protects the $y$-sets} if, for every subset $L$ of $[m-1]$ of size $y$, $L \supseteq E$ for some hyperedge $E$ of $\mathcal{H}$. The more fine-grained choice of $\Gamma$ we make is
\begin{equation}
    \Gamma = \{A \subseteq [m-1]: |A| \leq w \text{ and } E \not\subseteq A \text{ for all hyperedges } E \text{ of } \mathcal{H}\}.
\end{equation}
In other words, $\mathcal{H}$ ``protects'' certain points from being punctured by $\Gamma$. Note that $\Gamma$ is still a downset.

If $A$ is a set not in $\Delta$ of size $|A| \geq y$ with $A \subseteq [m-1]$, then $A$ is protected from the puncturing, since there is some hyperedge $E$ of $\mathcal{H}$ (of size $t$) for which $E \subseteq A$. Every set contained in $A$ that contains $E$ is protected from the puncture, of which there are $2^{|A|-|E|}$, and we establish that $\left|2^A\setminus \Gamma\right| \geq 2^{|A|-|E|}$. On the other hand, if $A$ is a minimal set not in $\Delta$ for which $A \not\subseteq [m-1]$, i.e., $m \in A$, then $A$ has at least $2^{|A|-1}$ subsets not in $\Gamma$, namely, those containing $m$, and we have $\left|2^A \setminus \Gamma\right| \geq 2^{|A|-1}$. Bounding the distance of the quantum code then becomes a case of arguing the lower bound on $\left|2^{A}\setminus \Gamma\right|$ over every set $A \notin \Delta$; see Section~\ref{sec:random_construction}.

For our choice of the hypergraph $\mathcal{H}$, we make a simple existence argument based on choosing hyperedges randomly. We find that doing so allows us to choose a relatively small number of hyperedges to protect the $y$-sets for a relatively large $y$. Because the number of hyperedges is relatively small, we retain a relatively large puncture set, and therefore a high-rate quantum code, but now with an improved distance in certain cases.

\subsection{Relation to Prior Works}

The first quantum code to support a transversal $T$ gate was the $[[15,1,3]]$ quantum (punctured) Reed-Muller code~\cite{knill1996threshold,steane1999quantum,MSD}. It was later realised that 3-dimensional topological codes could support non-Clifford transversal gates like the $T$ gate~\cite{bombin2007topological,kubica2015universal}. In 2012, Bravyi and Haah developed the triorthogonality framework~\cite{bravyi2012magic}. This allows one to construct codes for which the physical transversal $T$ gate, followed by some physical Clifford correction, executes the logical transversal $T$ gate: a weaker notion than that in this paper. Necessary and sufficient conditions for several related notions of T-gate transversality, including codespace preservation and specified logical actions, were subsequently developed~\cite{rengaswamy2020optimality,hu2022designing}. For example, a CSS-T code~\cite{rengaswamy2020optimality,camps2024toward,camps2024algebraic,camps2024binary,bodur2026schur} is a quantum CSS code for which the $T$ gate on every physical qubit only preserves the codespace, where there is no requirement on the corresponding logical action; in general the logical action may be the identity. Indeed,~\cite{Berardini2025AsymptoticallyGoodCSST} constructed the first asymptotically good CSS-T codes, and the logical action is the identity. \cite{Reddy2026AsymptoticallyGoodCSS} then constructed an asymptotically good family for which the physical transversal $T$ gate induces the Clifford operator $S^\dagger$ on every logical qubit. Further recent works~\cite{cao2026quantum,reddy2026realizing} have used novel methods to construct new families of codes for which transversal physical $T$ realises novel logical non-Clifford operations, for example $T$ and $T^\dagger$ on specified sets of logical qubits.

Bravyi and Haah's paper~\cite{bravyi2012magic} was also part of a line of work reducing the magic state distillation exponent $\gamma$~\cite{MSD,Meier2013MagicState,Jones2013Multilevel,Haah2017MagicState,hastings2018distillation,krishna2018towards,wills2024,golowich2025asymptotically,nguyen2024good,nguyen2025quantum}, eventually achieving $\gamma = 0$ with codes supporting transversal multi-qubit or qudit gates. Of these, the codes in~\cite{hastings2018distillation} support a transversal $T$ gate in the sense of this paper, up to the equivalence mentioned in Remark~\ref{rmk:T_gate_transversality}. Another key work from a similar time was that of Haah~\cite{haah2018towers}, and together these two papers constitute the best asymptotic families of quantum codes with transversal $T$ gates prior to this work. A work that derives different results, but is closely related to ours in some of its techniques, is~\cite{krishna2018magic}, which studies magic state distillation protocols by using punctured polar codes, that is, punctured decreasing monomial codes; we also use punctured decreasing monomial codes for this purpose in this paper. Many works are also studying non-Clifford transversal gates on quantum LDPC codes~\cite{golowich2025quantum,scruby2024quantum,zhu2025transversal,golowich2025near,zhu2025non,zhu2026non,kobayashi2026clifford,breuckmann2026cups,li2025poincar,li2026transversal}. Here, the primary focus is on transversal $CCZ$ gates because the theory is more mature than that for $T$ gates, and in particular the logical action in all of these constructions is some product of $CCZ$ gates, as far as we know.\footnote{For most of these constructions, the physical gate is transversal $CCZ$, and for all as far as we know the logical action is some product of $CCZ$ gates. Nevertheless, if the physical gate is transversal $CCZ$, one can concatenate with the $[[8,3,2]]$ quantum code, for which physical transversal $T$ implements logical $CCZ$, to obtain an LDPC code for which physical transversal $T$ implements the same logical product of $CCZ$ gates.}

Finally, further recent works have studied quantum codes with transversal non-Clifford gates in the finite-length regime. Notably,~\cite{JainAlbert2025Transversal} studies quantum codes supporting transversal $T$ gates in this setting, with some asymptotic families. Other works have studied and classified quantum codes with other notions of non-Clifford gate transversality, some allowing Clifford corrections, in the finite-length regime~\cite{nezami2022classification,cervia2025magic,Jacinto2026CompactMagicState,Singh2026BorrowedIdentities,Gong2026BinaryExtensionFields}, with the primary motivation being to construct or classify small magic state distillation protocols.

\section{Preliminaries}

\subsection{Classical Decreasing Monomial Codes}\label{subsec:params_and_duals}

Throughout, we consider the Boolean hypercube in $m$ variables, that is, $\mathbb{F}_2^m$. We identify points on the Boolean hypercube with subsets of $[m] \coloneq \{1, 2, \ldots, m\}$, that is, $\mathbb{F}_2^m \equiv 2^{[m]}$, in the natural way. Accordingly, we will think of a subset $\Delta \subseteq 2^{[m]}$ as both a set of subsets of $[m]$, and as a set of points in the Boolean hypercube. We work with the polynomial ring $\mathbb{F}_2[x_1, x_2, \ldots, x_m]/(x_1^2-x_1, x_2^2-x_2, \ldots, x_m^2-x_m)$, whose elements are the polynomials in $m$ variables, with coefficients in $\mathbb{F}_2$, for which each variable satisfies $x_i^2 = x_i$. The set of polynomials in this ring is then in one-to-one correspondence with the set of functions $f : \mathbb{F}_2^m \to \mathbb{F}_2$. In turn, these are in one-to-one correspondence with the subsets of $\mathbb{F}_2^m$, where a polynomial $f: \mathbb{F}_2^m \to \mathbb{F}_2$ is uniquely specified by its support $\supp(f)$. Given a point $A \in \mathbb{F}_2^m$, we write $x^A$ for the monomial corresponding to $A$, which is an element of the above polynomial ring.

We write $\ev: \mathbb{F}_2[x_1, x_2, \ldots, x_m]/(x_1^2-x_1, x_2^2-x_2, \ldots, x_m^2-x_m) \to \mathbb{F}_2^{2^m}$ for the evaluation map, which evaluates a function $f$ at all points in the hypercube $\mathbb{F}_2^m$. More formally, we identify the $2^m$ coordinates in $\mathbb{F}_2^{2^m}$ with bit strings $x \in \mathbb{F}_2^m$ (in some fixed ordering). Then, $\ev(f)$ is a bit string of length $2^m$, where for each $x \in \mathbb{F}_2^m$, the $x$'th entry of $\ev(f)$ is $f(x)$. Given two bit strings $x$ and $y$ of the same length, $x*y$ denotes the componentwise multiplication of $x$ and $y$, otherwise known as their overlap. We will make use of the fact that $\ev(f)*\ev(g) = \ev(fg)$.

\begin{definition}[Downset]
    A set of points in the hypercube $\Delta \subseteq 2^{[m]} \equiv \mathbb{F}_2^m$ is called downward-closed if, given $A \in \Delta$, and given some $B \in \mathbb{F}_2^m$ for which $B \subseteq A$, we have $B \in \Delta$. For the sake of brevity, we also call such a $\Delta$ a \textit{downset}, using the language of~\cite{srinivasan2019decoding}.
\end{definition}

\begin{definition}[(Decreasing) Monomial Codes]\label{def:decreasing_monomial_code}
    Given a non-empty set of points $\Delta \subseteq 2^{[m]}$, the binary affine monomial code corresponding to $\Delta$ is
    \begin{equation}
        C(\Delta) \coloneq \text{span}_{\mathbb{F}_2}\{\ev(x^A):A \in \Delta\}.
    \end{equation}
    If $\Delta$ is empty, then $C(\Delta) = \{0\}$. Since we only study binary affine monomial codes, we will sometimes simply call $C(\Delta)$ the monomial code corresponding to $\Delta$ for brevity. If $\Delta$ is a downset, the corresponding code is called a decreasing (binary affine) monomial code.
\end{definition}
It is easy to characterise the weight of a codeword corresponding to a monomial.
\begin{proposition}\label{prop:monomial_codeword_weight}
    Given $A \in 2^{[m]}$, the weight of the codeword $\ev(x^A)$ is
    \begin{equation}
        2^{m-|A|}.
    \end{equation}
    In particular, $\ev(x^A)$ has even Hamming weight if and only if $A \neq [m]$.
\end{proposition}
\begin{proof}
    The function $x^A$ is supported at a point in $\mathbb{F}_2^m$ if and only if all variables corresponding to $A$ take the value $1$; the remaining $m-|A|$ variables are arbitrary.
\end{proof}
\begin{lemma}[Parameters of Binary Monomial Codes]
    Given $\emptyset \neq \Delta \subseteq \mathbb{F}_2^m$, $C(\Delta)$ is a linear code of length $2^m$ with dimension $|\Delta|$. Its minimum distance is
    \begin{equation}
        d(C(\Delta)) = \min_{A \in \Delta} 2^{m-|A|}.
    \end{equation}
\end{lemma}
\begin{proof}
    The length of $C(\Delta)$ is immediate from its definition. Its dimension is also immediate since monomials are linearly independent functions, and the evaluation map is a linear bijection. For the distance, let $r \coloneq \max_{A \in \Delta}|A|$. Then, $C(\Delta)$ is contained in $\RM(m,r)$, which has minimum distance $2^{m-r}$, and so the claimed distance is indeed a lower bound. The upper bound follows from Proposition~\ref{prop:monomial_codeword_weight}.
\end{proof}
\begin{remark}
    It turns out that the parameters of binary monomial codes are much easier to prove than general $q$-ary codes, for which one must use the footprint bound from Gr\"obner basis theory~\cite{hoholdt1998or,geil2000footprints}.
\end{remark}
Next, we study the dual code to decreasing monomial codes. First, we make a definition.
\begin{definition}[Dual Downset]\label{def:dual_downset}
    Given a downset $\Delta \subseteq \mathbb{F}_2^m$, its dual downset is
    \begin{equation}
        \Delta^\perp \coloneq \{A \subseteq [m]: A \cup B \neq [m] \text{ for all } B \in \Delta\}.
    \end{equation}
\end{definition}
\begin{remark}\label{rmk:equivalent_dual_downset_characterisation}
    An equivalent, and sometimes convenient characterisation of the dual downset is
    \begin{equation}
        \Delta^\perp = \{A \subseteq [m]: \bar{A} \notin \Delta\}.
    \end{equation}
\end{remark}
\begin{lemma}[Dual of Decreasing Monomial Codes]\label{lem:monomial_code_dual}
    Given a downset $\Delta$, its dual $\Delta^\perp$ is a downset, and moreover the dual of the corresponding decreasing monomial code is the decreasing monomial code
    \begin{equation}
        C(\Delta)^\perp = C(\Delta^\perp).
    \end{equation}
\end{lemma}
\begin{remark}
    This fact is well-known, see for example Proposition 6 of~\cite{bardet2016algebraic}, although we use language here that we feel to be most appropriate for our purposes.
\end{remark}
\begin{proof}
    Given $A \in \Delta^\perp$, we have $A \cup B \neq [m]$ for all $B \in \Delta$. For all $C \subseteq A$, for all $B \in \Delta$, we have $C \cup B \subseteq A \cup B \neq [m]$, and so $\Delta^\perp$ is a downset.

    To see the code duality, we start by showing that $|\Delta^\perp| = 2^m - |\Delta|$. To see this, first notice the simple case $\Delta = \{(0, 0, \ldots, 0)\}$, we have $\Delta^\perp = 2^{[m]} \setminus \{(1, 1, \ldots, 1)\}$, in which case we have $|\Delta| = 1$ and $|\Delta^\perp| = 2^m-1$. Then, consider building up $\Delta$ by adding elements one-by-one, ensuring $\Delta$ is a downset at each step. Every time the size of $\Delta$ increases by one, the size of $\Delta^\perp$ decreases by one, and so we indeed have $|\Delta^\perp| = 2^m - |\Delta|$. 

    $|\Delta^\perp| = 2^m - |\Delta|$ implies $\dim C(\Delta)^\perp = \dim C(\Delta^\perp)$, and so we must only show one side is contained in the other. We conclude the proof by noting that if $A \in \Delta$ and $B \in \Delta^\perp$, we have
    \begin{align}
        \ev(x^A)\cdot \ev(x^B) &= |\ev(x^A) * \ev(x^B)| \\
        &= |\ev(x^Ax^B)|\\
        &= |\ev(x^{A\cup B})| \\
        &\equiv 0 \pmod 2.
    \end{align}
    In the first line, we recall the notation $*$ for the componentwise multiplication of two bit strings. Going into the second line, we have used the fact that, given any two functions $f,g$, we have $\ev(fg) = \ev(f) * \ev(g)$. Going into the final line, we have used the fact that $A \cup B \neq [m]$; see Proposition~\ref{prop:monomial_codeword_weight}.
\end{proof}
\subsection{Divisibility of Classical Codes}

Above, in Proposition~\ref{prop:monomial_codeword_weight}, we established the weight of any codeword corresponding to a monomial. It is much harder in general to characterise the weights of codewords corresponding to more general polynomials. However, we can get a handle on their properties using the Ward criteria, as follows.
\begin{definition}[$8$-Divisibility of a Binary Code]
    A binary linear code $C$ is called $8$-divisible if all its codewords have Hamming weight divisible by $8$.
\end{definition}
\begin{lemma}[Ward Criteria for Divisibility]\label{lemma:ward_criteria}
    Given a binary linear code $C$, we consider a basis for it; call the basis $c_1, c_2, \ldots, c_s$, where $s = \dim C$. $C$ being $8$-divisible is equivalent to the following conditions all being true simultaneously:
    \begin{enumerate}
        \item For each $\alpha \in [s]$, $|c_\alpha|$ is divisible by $8$;
        \item For each $\alpha,\beta \in [s]$ with $\alpha<\beta$, $|c_\alpha*c_\beta|$ is divisible by $4$;
        \item For each $\alpha,\beta,\gamma \in [s]$ with $\alpha < \beta < \gamma$, $|c_\alpha*c_\beta*c_\gamma|$ is divisible by $2$.
    \end{enumerate}
\end{lemma}
\begin{remark}
    Notice that if the Ward criteria hold for a particular basis of $C$, the code $C$ is $8$-divisible, and therefore the Ward criteria hold for all bases of $C$.
\end{remark}
The proof may be found in~\cite{ward1990weight}. Indeed, that reference proves the natural generalisation of the above statement for all $2^e$-divisibility. We provide a sketch here for completeness.
\begin{proof}
    Any codeword $c$ of $C$ is uniquely specified by some set of basis elements $B \subseteq [s]$ and equals
    \begin{equation}
        c = \sum_{b \in B} c_b,
    \end{equation}
    where $c = 0$ if $B = \emptyset$.
    The inclusion-exclusion formula then gives us
    \begin{equation}
        |c| \equiv \sum_{b \in B}|c_b| - 2\sum_{\{b,c\} \subseteq B} |c_b*c_c| + 4\sum_{\{b,c,d\} \subseteq B} |c_b*c_c*c_d| \pmod 8.
    \end{equation}
    If the Ward criteria hold for this basis at $e = 3$, then the code $C$ is $8$-divisible. 
    
    On the other hand, if the code $C$ is $8$-divisible, given a basis $c_1, c_2, \ldots, c_s$, their weights must be divisible by $8$. We then have that the weight of the overlap of any two/three basis elements is divisible by four/two from a similar inclusion/exclusion principle. For example,
    \begin{equation}
        2|c_1*c_2| = |c_1| + |c_2| - |c_1+c_2|,
    \end{equation}
    and the right-hand side is divisible by $8$.
\end{proof}

Using the above lemma, it is possible to provide simple sufficient conditions under which a (decreasing) monomial code is $2^e$-divisible. Since the transversality of $T$ gates is our main focus ($e=3$ in the previous lemma), we only show this case for simplicity.

\begin{proposition}[$8$-Divisibility of (Decreasing) Monomial Codes]\label{prop:monomial_divisibility}
    Consider a set of points $\Delta \subseteq 2^{[m]}$, the binary affine monomial code corresponding to $\Delta$, $C(\Delta)$, is $8$-divisible if the following conditions hold simultaneously:
    \begin{enumerate}
        \item For each $A \in \Delta$, $\left|A\right| \leq m-3$;
        \item For each $A,B \in \Delta$, $\left|A\cup B\right| \leq m-2$;
        \item For each $A,B,C \in \Delta$, $\left|A \cup B \cup C\right| \leq m-1$.
    \end{enumerate}
\end{proposition}
\begin{proof}
    For $A \in \Delta$, $\ev(x^A)$ form a basis for the code $C(\Delta)$. Noting that we have $\ev(x^A)*\ev(x^B) = \ev(x^{A\cup B})$ and $\ev(x^A)*\ev(x^B)*\ev(x^C) = \ev(x^{A \cup B \cup C})$, the statement follows from Proposition~\ref{prop:monomial_codeword_weight} and Lemma~\ref{lemma:ward_criteria}.
\end{proof}
\begin{remark}
    While we only proved sufficiency of the conditions in Proposition~\ref{prop:monomial_divisibility}, the conditions turn out to be necessary and sufficient.
\end{remark}

\subsection{Classical Weighted Reed-Muller Codes}
Let us now introduce a generalisation of Reed-Muller codes which will be important throughout the paper.

\begin{definition}[Weighted Reed-Muller Codes]\label{def:weighted_RM}
    Let $\boldsymbol{s} = (s_1, s_2, \ldots, s_m)$ be a vector of positive integers. Given a bit string $A \in \mathbb{F}_2^m$, or equivalently a subset $A \subseteq [m]$, we define a generalised weight
    \begin{equation}\label{eq:gen_weight_first_def}
        |A|_{\boldsymbol{s}} = \sum_{i=1}^mA_is_i,
    \end{equation}
    viewing $A$ as a bit string, or
    \begin{equation}
        |A|_{\boldsymbol{s}} = \sum_{i \in A}s_i,
    \end{equation}
    viewing $A$ as a subset of $[m]$.
    Given a choice of $\boldsymbol{s}$, and a positive integer $r$, we have a classical weighted Reed-Muller code~\cite{sorensen1992weighted,geil2013weighted} defined as $C(\mathcal{W}_{\boldsymbol{s},r})$, where $\mathcal{W}_{\boldsymbol{s},r} \coloneq \{A \subseteq [m]: |A|_{\boldsymbol{s}} \leq r\}$, in the notation of Definition~\ref{def:decreasing_monomial_code}.
\end{definition}
The following is a known result~\cite{sorensen1992weighted}, but we prove it here for completeness.
\begin{proposition}[Dual of a Weighted Reed-Muller Code]\label{prop:dual_weighted_rm} Consider a classical weighted Reed-Muller code $C(\mathcal{W}_{\boldsymbol{s},r})$ as in Definition~\ref{def:weighted_RM}. Let $S = \sum_{i=1}^ms_i$, and assume $S \geq r+1$. The dual of $C(\mathcal{W}_{\boldsymbol{s},r})$ is $C(\mathcal{W}_{\boldsymbol{s},S-r-1})$.
\end{proposition}
\begin{proof}
    Since $\mathcal{W}_{\boldsymbol{s},r}$ is a downset, by Lemma~\ref{lem:monomial_code_dual}, it is enough to prove that $\mathcal{W}_{\boldsymbol{s},r}^\perp = \mathcal{W}_{\boldsymbol{s},S-r-1}$. To do this, we use the alternative characterisation of the dual downset, as in Remark~\ref{rmk:equivalent_dual_downset_characterisation}. For any set $B \subseteq [m]$, we have $B \notin \mathcal{W}_{\boldsymbol{s},r}$ if and only if $|B|_{\boldsymbol{s}} > r$. Moreover, for any $B \subseteq [m], |B|_{\boldsymbol{s}}> r$ if and only if $|\bar{B}|_{\boldsymbol{s}}\leq S-r-1$, and the result follows.
\end{proof}
We give sufficient conditions for weighted Reed-Muller codes to be $8$-divisible.
\begin{lemma}[$8$-Divisibility of Weighted Reed-Muller Codes]\label{lem:divis_weighted_RM}
    Consider a classical weighted Reed-Muller code $C(\mathcal{W}_{\boldsymbol{s},r})$ as in Definition~\ref{def:weighted_RM}. Let $S = \sum_{i=1}^ms_i$, $s^{(1)}$ be the largest of the $s_i$, breaking ties arbitrarily, and let $s^{(2)}$ be the second-largest of the $s_i$, again breaking ties arbitrarily, and where we may have $s^{(2)} = s^{(1)}$. Suppose
    \begin{align}
        3r &< S\\
        2r &< S-s^{(1)}\\
        r &< S-s^{(1)}-s^{(2)}.
    \end{align}
    Then, $C(\mathcal{W}_{\boldsymbol{s},r})$ is $8$-divisible.
\end{lemma}
\begin{proof}
    Since $\mathcal{W}_{\boldsymbol{s},r}$ is a downset, it suffices to check that it satisfies the conditions of Proposition~\ref{prop:monomial_divisibility}. Given $A,B,C \in \mathcal{W}_{\boldsymbol{s},r}$, we have $|A\cup B \cup C|_{\boldsymbol{s}} \leq |A|_{\boldsymbol{s}} + |B|_{\boldsymbol{s}} + |C|_{\boldsymbol{s}} \leq 3r < S$. Since $[m]$ is the unique set of weight $S$, we cannot have $A \cup B \cup C = [m]$. Moreover, we have
    \begin{equation}
        \min_{\substack{D \subseteq [m]\\|D| = m-1}}|D|_{\boldsymbol{s}} = S-s^{(1)}.
    \end{equation}
    For every $A,B \in \mathcal{W}_{\boldsymbol{s},r}$, we have $|A\cup B|_{\boldsymbol{s}} \leq 2r < S-s^{(1)}$, and so $A \cup B$ cannot contain a subset of $[m]$ of size $m-1$; in other words, $|A \cup B| \leq m-2$. Finally,
    \begin{equation}
        \min_{\substack{D \subseteq [m]\\|D| = m-2}} |D|_{\boldsymbol{s}} = S-s^{(1)}-s^{(2)},
    \end{equation}
    and it follows that every $A \in \mathcal{W}_{\boldsymbol{s},r}$ has $|A| \leq m-3$.
\end{proof}

\subsection{Quantum CSS Codes from X-Generator Matrices}\label{subsec:prelims_X_gen}

We now review a standard construction for developing quantum CSS codes, where $X$ stabilisers and $X$ logicals are developed from a single classical code, often with algebraic properties, in order to ensure the quantum code has some desirable transversal gates. See~\cite{nielsen2000quantum} for the preliminary material on quantum CSS codes.

Let $C$ be some classical linear code. We write down a generator matrix for $C$, call it $\hat{G}$, that is, a matrix whose rows form a basis for $C$. Consider some set of columns of $\hat{G}$ that are linearly independent. We permute the columns of $\hat{G}$ such that these become the $k$ left-most columns of $\hat{G}$, and after row operations, we write $\hat{G}$ into the form
\begin{equation}\label{eq:semi_systematic_form}
    \hat{G} = \begin{pmatrix}
        I_k & G_1\\
        0 & G_0
    \end{pmatrix},
\end{equation}
where $I_k$ is the $k \times k$ identity matrix. It is convenient to write
\begin{equation}\label{eq:X_stab_mat}
    G \coloneq \begin{pmatrix}
        G_1\\
        G_0
    \end{pmatrix},
\end{equation}
and we let $\mathcal{G}, \mathcal{G}_1$, and $\mathcal{G}_0$ denote the row spans of the matrices $G$, $G_1$ and $G_0$, respectively. We then form a quantum CSS code $\mathcal{Q}$, whose $X$ stabilisers are given by $\mathcal{G}_0$, and whose $Z$ stabilisers are given by $\mathcal{G}^\perp$. We then have the following.
\begin{proposition}\label{prop:quantum_code_logical_properties}
    Let $C$ be a self-orthogonal classical code, that is, $C \subseteq C^\perp$. The resulting quantum CSS code $\mathcal{Q}$ from the above construction satisfies the following:
    \begin{enumerate}
        \item $\mathcal{Q}$ has $k$ logical qubits;
        \item The distance of $\mathcal{Q}$ is equal to its $Z$-distance, which is
        \begin{equation}
            d_Z(\mathcal{Q}) = \min_{f \in \mathcal{G}_0^\perp\setminus\mathcal{G}^\perp}|f|;
        \end{equation}
        \item The rows of $G_1$ may be used as a canonical basis of non-trivial $X$ and $Z$-logical operators.
    \end{enumerate}
\end{proposition}
\begin{remark}
    This proof is quite standard; see Lemma 1 of~\cite{bravyi2012magic}, although we include it here for completeness.
\end{remark}
\begin{proof}
    Let $G$ have $s$ rows and $n$ columns. Let the rows of $G$ be $(g_i)_{i=1}^s$. The self-orthogonality of $C$, as well as the structure of the matrix in Equation~\eqref{eq:semi_systematic_form} imply that
    \begin{equation}\label{eq:g_row_overlaps}
        g_i \cdot g_j = \begin{cases}
            1 &\text{ if } 1 \leq i = j \leq k\\
            0 &\text{ otherwise.}
        \end{cases}
    \end{equation}
    We then consider any linear relation
    \begin{equation}
        \sum_{i=1}^kc_ig_i = \sum_{i=k+1}^sc_ig_i
    \end{equation}
    for any $(c_i)_{i=1}^s \in \mathbb{F}_2^s$. Taking the inner product with $g_j$, for each $j = 1, \ldots, k$ in turn, yields $c_j = 0$ for each $j = 1, \ldots, k$. This tells us that $\mathcal{G}_0 \cap \mathcal{G}_1 = 0$, and that the rows of $G_1$ are linearly independent. The number of logical qubits of $\mathcal{Q}$ is therefore
    \begin{equation}
        \dim\mathcal{G} - \dim\mathcal{G}_0 = k.
    \end{equation}
    Next, let us show that $\mathcal{G}_0^\perp = \mathcal{G}_1 + \mathcal{G}^\perp$. Given Equation~\eqref{eq:g_row_overlaps}, we have $\mathcal{G}_1 \cap \mathcal{G}^\perp = 0$, and therefore
    \begin{align}
        \dim(\mathcal{G}_1+\mathcal{G}^\perp) &= \dim\mathcal{G}_1 + \dim\mathcal{G}^\perp\\
        &= k + n - \dim\mathcal{G}\\
        &= k + n - k - \dim\mathcal{G}_0\\
        &=n-\dim\mathcal{G}_0\\
        &= \dim\mathcal{G}_0^\perp.
    \end{align}
    Moreover, $\mathcal{G}_1 \subseteq \mathcal{G}_0^\perp$ by Equation~\eqref{eq:g_row_overlaps}, and $\mathcal{G}^\perp \subseteq \mathcal{G}_0^\perp$ by definition, giving $\mathcal{G}_1 + \mathcal{G}^\perp \subseteq \mathcal{G}_0^\perp$, and we indeed have $\mathcal{G}_0^\perp = \mathcal{G}_1 + \mathcal{G}^\perp$.

    This fact, along with $\mathcal{G}_1 \cap \mathcal{G}^\perp = 0$ give us
    \begin{equation}
        \mathcal{G}_0^\perp \setminus\mathcal{G}^\perp = \{x+y:x \in \mathcal{G}_1\setminus\{0\}, y \in \mathcal{G}^\perp\}.
    \end{equation}
    However, because $\mathcal{G}^\perp \supseteq \mathcal{G}_0$, we have
    \begin{align}
        \mathcal{G}_0^\perp \setminus \mathcal{G}^\perp \supseteq \{x+y: x \in \mathcal{G}_1\setminus\{0\}, y \in \mathcal{G}_0\}
        = \mathcal{G}\setminus \mathcal{G}_0,
    \end{align}
    where in the equality we have used $\mathcal{G}_0 \cap \mathcal{G}_1 = 0$. Since the $X$ and $Z$-distances are
    \begin{align}
        d_X &= \min_{f \in \mathcal{G}\setminus\mathcal{G}_0}|f|\\
        d_Z &= \min_{f \in \mathcal{G}_0^\perp \setminus \mathcal{G}^\perp}|f|,
    \end{align}
    respectively, we have $d_X \geq d_Z$, and the distance of the quantum code is equal to its $Z$ distance, which has the given expression.

    Finally, the rows of $G_1$ satisfy the correct orthogonality relation, both with themselves, and with the stabilisers to form a canonical basis of both $X$ and $Z$ non-trivial logical operators. Moreover, $\mathcal{G}_0 \cap \mathcal{G}_1 = 0$ and $\mathcal{G} = \mathcal{G}_0 + \mathcal{G}_1$ implies that $(g_i + \mathcal{G}_0)_{i=1}^k$ forms a basis for $\mathcal{G}/\mathcal{G}_0$, whereas $\mathcal{G}_0^\perp = \mathcal{G}_1 + \mathcal{G}^\perp$ and $\mathcal{G}_1 \cap \mathcal{G}^\perp = 0$ imply that $(g_i +\mathcal{G}^\perp)_{i=1}^k$ forms a basis for $\mathcal{G}_0^\perp/\mathcal{G}^\perp$.
\end{proof}
We omit the proof of the following proposition linking the divisibility of the classical code $C$ to its transversality properties, since it is a standard calculation~\cite{hastings2018distillation,bravyi2012magic}. In the statement, we consider the one-qubit gate which is in exactly the third level of the Clifford hierarchy~\cite{gottesman1999demonstrating}:
\begin{equation}
    T = \begin{pmatrix}
        1 & 0\\
        0 & e^{\frac{i\pi}{4}}
    \end{pmatrix}.
\end{equation}
\begin{proposition}\label{prop:quantum_code_transversality_from_divisibility}
    Suppose that the classical linear code $C$ in the above construction is $8$-divisible, and suppose that we make a choice of logical basis for the quantum CSS code $\mathcal{Q}$ as given in Proposition~\ref{prop:quantum_code_logical_properties}. The physical operator acting transversally with $T$ on each physical qubit executes the operator $T^\dagger$ on each logical qubit transversally, that is,
    \begin{equation}
        T^{\otimes n} = \overline{(T^\dagger)^{\otimes k}}.
    \end{equation}
    In short, $\mathcal{Q}$ supports a transversal $T$ gate (without Clifford corrections).
\end{proposition}
\begin{remark}
    Given a quantum CSS code with a choice of logical basis for which $T^{\otimes n} = \overline{(T^\dagger)^{\otimes k}}$, one can change its logical basis so that $T^{\otimes n} = \overline{T^{\otimes k}}$; see Remark~\ref{rmk:T_gate_transversality}.
\end{remark}

\section{Quantum Codes from Decreasing Monomial Codes Punctured at a Downset}

\subsection{The Distance of Classical Decreasing Monomial Codes Punctured at a Downset}

To construct improved quantum codes with transversal $T$ gates, we will need to be able to establish the distance of classical decreasing monomial codes punctured at a downset; we find a closed-form expression for this quantity in this section.
\begin{definition}[Punctured Monomial Codes]\label{def:punc_monomial_codes}
    Given some $\Delta \subseteq 2^{[m]}$, the coordinates of the code $C(\Delta)$ may themselves be identified with the coordinates of the Hamming cube $\mathbb{F}_2^m \equiv 2^{[m]}$. Given some set $\Gamma \subseteq \mathbb{F}_2^m$, we let $C(\Delta, \Gamma)$ be the monomial code $C(\Delta)$ punctured at the points $\Gamma$, that is,
    \begin{equation}
        C(\Delta, \Gamma) \coloneq \punc_{\Gamma}C(\Delta).
    \end{equation}
    This is a code of length $2^m - |\Gamma|$. When $\Delta$ and $\Gamma$ are both downsets, we call $C(\Delta, \Gamma)$ a decreasing monomial code punctured at a downset.
\end{definition}
The following lemma on the distance of a decreasing monomial code punctured at a downset is new as far as we know.
\begin{lemma}[Distance of a Decreasing Monomial Code Punctured at a Downset]\label{lem:monomial_puncture_distance}
    Let $\Delta, \Gamma \subseteq 2^{[m]}$ be downsets. Suppose that
    \begin{equation}\label{eq:non_annihilate_assumption}
        \bar{A} \in\bar{\Gamma} \text{ for every } A \in \Delta.\footnote{For clarity, note that $A \subseteq [m]$, and so $\bar{A} \coloneq [m]\setminus A$, whereas $\Gamma \subseteq 2^{[m]}$, and so $\bar{\Gamma} \coloneq 2^{[m]} \setminus \Gamma$.}
    \end{equation}
    Let $C(\Delta, \Gamma)$ be the corresponding decreasing monomial code punctured at a downset. The minimum distance of $C(\Delta, \Gamma)$ is
    \begin{equation}
        d(C(\Delta, \Gamma)) = \min_{A \in \Delta}|\{V \subseteq \bar{A}: V \in \bar{\Gamma}\}|.
    \end{equation}
\end{lemma}
\begin{remark}
    Note that the assumption $\bar{A} \in \bar{\Gamma}$ for every $A \in \Delta$ implies that $|\{V \subseteq \bar{A}: V \in \bar{\Gamma}\}|$ is positive for every $A \in \Delta$, and thus we are not claiming a distance of zero, which would be meaningless. Moreover, the condition of Equation~\eqref{eq:non_annihilate_assumption} is a natural one, because it says that in our puncturing, we are not removing an entire codeword of the original code $C(\Delta)$; in other words, we are only considering the situation that puncturing is an injective code operation; we referred to this in the discussion in Section~\ref{subsubsec:discussion_decreasing_monomial_downset} as the non-annihilation condition, because it says that no codeword is ``annihilated'' by the puncturing.
\end{remark}
\begin{proof}
    First, we show the distance upper bound. For any $A \in \Delta$, consider the polynomial
    \begin{equation}
        f = \prod_{i \in A}(1+x_i).
    \end{equation}
    By expanding out the product, this polynomial is equal to
    \begin{equation}
        f = \sum_{B \subseteq A} x^B.
    \end{equation}
    Because $\Delta$ is a downset, the evaluation of this polynomial is a member of $C(\Delta)$. On $\mathbb{F}_2^m$, the polynomial evaluates to $1$ if and only if $x_i = 0$ for all $i \in A$, meaning that its support is exactly the set of points with the variables in $A$ set to $0$, and the variables in $\bar{A}$ allowed to vary. In other words, the support of $f$ is exactly
    \begin{equation}
        \supp(f) = \{V \subseteq \bar{A}\} \subseteq \mathbb{F}_2^m.
    \end{equation}
    Puncturing the points $\Gamma$, the support of $f$ is now
    \begin{equation}
        \text{supp}(f|_{\bar{\Gamma}}) = \{V \subseteq \bar{A}: V \in \bar{\Gamma}\},
    \end{equation}
    which is non-empty by the assumption in Equation~\eqref{eq:non_annihilate_assumption}.
    This means that, for every $A \in \Delta$, we have
    \begin{equation}
        d(C(\Delta, \Gamma)) \leq |\{V \subseteq \bar{A}: V \in \bar{\Gamma}\}|,
    \end{equation}
    and so
    \begin{equation}
        d(C(\Delta, \Gamma)) \leq \min_{A \in \Delta}|\{V \subseteq \bar{A}: V \in \bar{\Gamma}\}|.
    \end{equation}
    For the distance lower bound, let us consider any non-zero polynomial
    \begin{equation}
        f = \sum_{B \in \Delta}c_Bx^B: c_B \in \mathbb{F}_2.
    \end{equation}
    We begin by lower bounding the size of the support of $f$ on the whole Boolean hypercube $\mathbb{F}_2^m$, before considering the size of its support on $\bar{\Gamma}$.
    Let $A$ be any maximal monomial in the support of $f$, that is, let $A$ be any element in $\Delta$ for which $c_A = 1$ and for which $B \in \Delta$ with $A \subsetneq B \implies c_B = 0$. For every $V \subseteq \bar{A}$, we define a set of points $\mathcal{S}_V \subseteq \mathbb{F}_2^m$:
    \begin{equation}
        \mathcal{S}_V \coloneq \{U \sqcup V: U \subseteq A\}.
    \end{equation}
    We will show that for every $V$, $f$ must be supported at some point in $\mathcal{S}_V$. Since the $\mathcal{S}_V$ are disjoint, and there are $2^{m-|A|}$ possible choices for $V$, the weight of $f$ must be at least $2^{m-|A|}$. Indeed, for any $V$, consider the sum 
    \begin{align}
        \sum_{y \in \mathcal{S}_V}f(y) &= \sum_{y \in \mathcal{S}_V}\sum_{B \in \Delta}c_Bx^B(y)\\
        &=\sum_{B \in \Delta}c_B\sum_{y \in \mathcal{S}_V}x^B(y).
    \end{align}
    Given $B \in 2^{[m]}$ and $y \in \mathbb{F}_2^m$, we have $x^B(y) = 1$ if and only if $B \subseteq y$, where we recall the natural identification of $2^{[m]}$ with $\mathbb{F}_2^m$ as in Subsection~\ref{subsec:params_and_duals}. For each $B\in \Delta$ for which $c_B = 1$, we will now evaluate $\sum_{y \in \mathcal{S}_V}x^B(y)$. We first note that if $B \setminus A \not\subseteq V$ then $x^B(y) = 0$ for all $y \in \mathcal{S}_V$, because $y \setminus A = V$ for all $y \in \mathcal{S}_V$. On the other hand, suppose that $B \setminus A \subseteq V$, and suppose for the moment that $B \neq A$. Since $A$ was chosen to be maximal among the sets for which $c_A = 1$, the only possibility is that $B \cap A \neq A$. There are therefore an even number of points $y \in \mathcal{S}_V$ for which $B \subseteq y$, meaning that $\sum_{y \in \mathcal{S}_V} x^B(y) = 0$. The only remaining possibility is $B = A$, for which we have $\sum_{y \in \mathcal{S}_V} x^B(y) = 1$, and we conclude that
    \begin{equation}
        \sum_{y \in \mathcal{S}_V}f(y) = c_A = 1.
    \end{equation}
    This implies that $f$ is supported somewhere in $\mathcal{S}_V$, as required.

    We have established that $f$ has weight at least $2^{m-|A|}$ on the whole Boolean hypercube by establishing that it is supported somewhere in the $2^{m-|A|}$ disjoint sets $\mathcal{S}_V$, for each $V \subseteq \bar{A}$. Given any $V \subseteq \bar{A}$ for which $V \in \bar{\Gamma}$, we have $\mathcal{S}_V \subseteq \bar{\Gamma}$, which follows from the fact that $\Gamma$ is a downset. We conclude that the weight of $f$ on $\bar{\Gamma}$ is at least the number of subsets $V \subseteq \bar{A}$ for which $V \in \bar{\Gamma}$, which is at least the claimed distance.
\end{proof}
\begin{example}\label{example:hh_first}
    The choices in Hastings and Haah~\cite{hastings2018distillation} can be reproduced with the downsets
    \begin{align}
        \Delta &= \{A \subseteq [m]: |A| \leq r\}\\
        \Gamma &= \{A \subseteq [m]: |A| \leq w\}
    \end{align}
    for some integers $0\leq w<r<m/2$. Note that making this choice of $\Delta$ means that the un-punctured code $C(\Delta)$ in this case is simply the Reed-Muller code $\text{RM}(m,r)$. Our formula for the distance of the resulting punctured Reed-Muller code $C(\Delta, \Gamma)$ reproduces their computed formula, which is $\sum_{i=w+1}^{m-r}\left(\begin{smallmatrix}
        m-r\\i
    \end{smallmatrix}\right)$. Interestingly, they establish this particular formula using a different, inductive proof; see Lemma $1$ in~\cite{hastings2018distillation}.
\end{example}

\subsection{Classical to Quantum Codes}

We will now instantiate the construction of Section~\ref{subsec:prelims_X_gen} with the classical code $C(\Delta)$, and a puncture set $\Gamma$.

Let $\Delta, \Gamma \subseteq 2^{[m]}$ be downsets. Write a generator matrix $\hat{G}(\Delta)$ for the code $C(\Delta)$, that is, $\hat{G}(\Delta) \in \mathbb{F}_2^{|\Delta| \times 2^m}$, and the rows of $\hat{G}(\Delta)$ form a basis for $C(\Delta)$. Permute the columns of $\hat{G}(\Delta)$ such that the columns corresponding to $\Gamma$ are the left-most $|\Gamma|$ columns.
\begin{claim}\label{claim:column_rank}
    Suppose $\Gamma \subseteq \Delta$. Then, the columns of $\hat{G}(\Delta)$ corresponding to $\Gamma$ are linearly independent.
\end{claim}
\begin{proof}
    Choose the first $|\Gamma|$ rows of $\hat{G}(\Delta)$ to correspond to the monomials $\{x^\gamma: \gamma \in \Gamma\}$. Consider the $|\Gamma| \times |\Gamma|$ top-left sub-matrix of $\hat{G}(\Delta)$. We may further pick an order for the monomials of $\Gamma$, say, $\Gamma = \{\gamma_1, \gamma_2, \ldots, \gamma_{|\Gamma|}\}$ such that if $j > i$ then $\gamma_j \not \subseteq \gamma_i$ (this may be achieved by simply writing $\emptyset \in \Gamma$ first, then all of the sets $\gamma \in \Gamma$ of size one, then all of the sets of size two, and so on). We have that $x^{\gamma_j}(\gamma_i) = 0$ if $\gamma_j \not\subseteq \gamma_i$, and $x^{\gamma_i}(\gamma_i) = 1$, and therefore this $|\Gamma| \times |\Gamma|$ sub-matrix is upper-triangular with ones on the diagonal. It follows that this sub-matrix is full-rank, and thus that the first $|\Gamma|$ columns of $\hat{G}(\Delta)$ are linearly independent.
\end{proof}
Because of Claim~\ref{claim:column_rank}, under the assumption that $\Gamma \subseteq \Delta$, we may perform row operations on $\hat{G}(\Delta)$ such that its left-most $|\Gamma|$ columns take the form
\begin{equation}
    \begin{pmatrix}
        I_{|\Gamma|}\\
        0
    \end{pmatrix},
\end{equation}
where $I_{|\Gamma|}$ is the $|\Gamma| \times |\Gamma|$ identity matrix. We write $\hat{G}(\Delta)$ as
\begin{equation}
    \hat{G}(\Delta) = \begin{pmatrix}
        I_{|\Gamma|} & G_1(\Delta, \Gamma)\\
        0 & G_0(\Delta, \Gamma)
    \end{pmatrix},
\end{equation}
where we have defined matrices $G_1(\Delta, \Gamma)$ and $G_0(\Delta, \Gamma)$. We write
\begin{equation}
    G(\Delta, \Gamma) \coloneq \begin{pmatrix}
        G_1(\Delta, \Gamma)\\
        G_0(\Delta, \Gamma)
    \end{pmatrix}.
\end{equation}
\begin{definition}[Quantum Code from a Decreasing Monomial Code Punctured at a Downset]
    Given downsets $\Delta, \Gamma \subseteq \mathbb{F}_2^m$ such that $\emptyset \in \Gamma \subseteq \Delta$, with the notation of the matrices in Equations~\eqref{eq:semi_systematic_form} and~\eqref{eq:X_stab_mat}, let $\mathcal{G}_1(\Delta, \Gamma)$ be the row span of $G_1(\Delta, \Gamma)$, let $\mathcal{G}_0(\Delta, \Gamma)$ be the row span of $G_0(\Delta, \Gamma)$, and let $\mathcal{G}(\Delta, \Gamma)$ be the row span of $G(\Delta, \Gamma)$. The quantum CSS code corresponding to the downsets $\Delta, \Gamma$ is denoted $\mathcal{Q}(\Delta, \Gamma)$. It is the quantum CSS code with $X$-stabiliser space $\mathcal{G}_0(\Delta, \Gamma)$, and $Z$-stabiliser space $\mathcal{G}(\Delta, \Gamma)^\perp$; see Section~\ref{subsec:prelims_X_gen}.
\end{definition}
\begin{proposition}\label{prop:downsets_to_quantum_code}
    Consider downsets $\Delta, \Gamma \subseteq \mathbb{F}_2^m$ for which $\emptyset \in \Gamma \subseteq \Delta$ and for which $\Delta \subseteq \Delta^\perp$, using the notation of Definition~\ref{def:dual_downset}. Then, $\mathcal{Q}(\Delta, \Gamma)$ has $|\Gamma|$ logical qubits, and moreover its distance is equal to its $Z$-distance, which is in turn
    \begin{equation}\label{eq:quantum_code_distance}
        d_Z(\mathcal{Q}(\Delta, \Gamma)) = \min_{A \in \Delta^\perp}|\{V \subseteq \bar{A}:V \in \bar{\Gamma}\}|.
    \end{equation}
    Moreover, if $C(\Delta)$ is an $8$-divisible classical code, the quantum code $\mathcal{Q}(\Delta, \Gamma)$ has a transversal $T$ gate (without Clifford corrections), with the natural choice of logical basis as in Proposition~\ref{prop:quantum_code_logical_properties}.
\end{proposition}
\begin{proof}
    $\Delta \subseteq \Delta^\perp$ implies that $C(\Delta) \subseteq C(\Delta^\perp) = C(\Delta)^\perp$, i.e., the code $C(\Delta)$ is self-orthogonal. The fact that the code has $|\Gamma|$ logical qubits, and its distance is equal to its $Z$-distance, follows from Proposition~\ref{prop:quantum_code_logical_properties}. Its $Z$-distance is then
    \begin{equation}
        d_Z = \min_{f \in \mathcal{G}_0(\Delta, \Gamma)^\perp \setminus \mathcal{G}(\Delta, \Gamma)^\perp}|f| \geq \min_{f \in \mathcal{G}_0(\Delta, \Gamma)^\perp \setminus\{0\}}|f|,
    \end{equation}
    the classical distance of the code $\mathcal{G}_0(\Delta, \Gamma)^\perp$. Since $\mathcal{G}_0(\Delta, \Gamma)$ is the classical code $C(\Delta)$ shortened at the points $\Gamma$, $\mathcal{G}_0(\Delta, \Gamma)^\perp$ is the classical code $C(\Delta)^\perp = C(\Delta^\perp)$ punctured at the points $\Gamma$, which we denoted $C(\Delta^\perp, \Gamma)$ in Definition~\ref{def:punc_monomial_codes}. Given $A \in \Delta^\perp$, we must have $\bar{A} \in \bar{\Gamma}$, since if $\bar{A} \in \Gamma$, then $\Gamma \subseteq \Delta$ implies $\bar{A} \in \Delta$: a contradiction. The fact that $d_Z$ is lower-bounded by the claimed expression then follows from Lemma~\ref{lem:monomial_puncture_distance}.

    Finally, let us show that the $Z$-distance is actually equal to the claimed expression, i.e., the inequality is an equality. Let $A \in \Delta^\perp$ be such that
    \begin{equation}
        |\{V \subseteq \bar{A}: V \in \bar{\Gamma}\}| = \min_{A' \in \Delta^\perp} |\{V \subseteq \bar{A'}:V \in \bar{\Gamma}\}|,
    \end{equation}
    and consider
    \begin{equation}
        f = \prod_{i \in A}(1+x_i),
    \end{equation}
    for which $\ev(f)$ is in $C(\Delta^\perp) = C(\Delta)^\perp$ since $\Delta^\perp$ is a downset. Its evaluation vector restricted to $\bar{\Gamma}$, $\ev(f|_{\bar{\Gamma}})$, has weight $|\{V \subseteq \bar{A}: V \in \bar{\Gamma}\}|$, and lies in $\punc_{\Gamma}C(\Delta^\perp) = \mathcal{G}_0(\Delta, \Gamma)^\perp$, and thus corresponds to an undetectable $Z$ operator. To see it is not a $Z$-stabiliser, we must show that $\ev(f|_{\bar{\Gamma}})$ does not lie in $\mathcal{G}(\Delta, \Gamma)^\perp$, which we do by showing that it is not orthogonal to some row of $G$. Given $\emptyset \in \Gamma$,\footnote{Since $\Gamma$ is a downset, saying $\emptyset \in \Gamma$ is equivalent to saying $\Gamma$ is non-empty.} looking at the form of $\hat{G}(\Delta)$ in Equation~\eqref{eq:semi_systematic_form}, there is a function $h$, whose evaluation over $\mathbb{F}_2^m$ is in $C(\Delta)$, whose support on $\bar{\Gamma}$ is equal to a row of $G_1(\Delta, \Gamma)$, and whose support on $\Gamma$ has weight one: supported exactly at $\emptyset$. Since $f(\emptyset) = 1$, we have
    \begin{align}
        0 &= f(\emptyset)\cdot h(\emptyset) + \ev(f|_{\bar{\Gamma}})\cdot\ev(h|_{\bar{\Gamma}})\\
        &= 1 + \ev(f|_{\bar{\Gamma}})\cdot\ev(h|_{\bar{\Gamma}}),
    \end{align}
    and so $1 = \ev(f|_{\bar{\Gamma}})\cdot\ev(h|_{\bar{\Gamma}})$, meaning that indeed $\ev(f|_{\bar{\Gamma}})$ is not orthogonal to some row of $G_1(\Delta, \Gamma)$.

    The transversality of the $T$ gate follows from Proposition~\ref{prop:quantum_code_transversality_from_divisibility}.
\end{proof}

\begin{example}
    The choices for the quantum code in Hastings and Haah~\cite{hastings2018distillation} can be reproduced with the choices of the downsets $\Delta$ and $\Gamma$ as in Example~\ref{example:hh_first}. It is easy to see that the dual downset to $\Delta$ is
    \begin{equation}
        \Delta^\perp = \{A \subseteq [m]: |A| \leq m-r-1\}.
    \end{equation}
    Then, as long as we choose $w < r$ (to get $\Gamma \subseteq \Delta$) and $2r < m$ (to get $r \leq m-r-1$, and thus $\Delta \subseteq \Delta^\perp$), Proposition~\ref{prop:downsets_to_quantum_code} gives us that the quantum code $\mathcal{Q}(\Delta, \Gamma)$ has a number of logical qubits
    \begin{equation}
        |\Gamma| = \begin{pmatrix}
            m\\ \leq w
        \end{pmatrix} \coloneq \sum_{i=0}^w\begin{pmatrix}
            m\\i
        \end{pmatrix},
    \end{equation}
    which reproduces~\cite{hastings2018distillation}.
    Moreover, the distance of $\mathcal{Q}(\Delta, \Gamma)$ is then
    \begin{equation}
        \min_{A \in \Delta^\perp} |\{V \subseteq \bar{A}: V \notin \Gamma\}|.
    \end{equation}
    Considering any $A \in \Delta^\perp$, $A$ is a set of size $|A| \leq m-r-1$, and $\bar{A}$ is a set of size $m-|A|$. The number of its subsets that are then not in $\Gamma$ is then
    \begin{equation}
        \begin{pmatrix}
            m-|A|\\>w
        \end{pmatrix} \coloneq \sum_{i=w+1}^{m-|A|}\begin{pmatrix}
            m-|A|\\i
        \end{pmatrix}.
    \end{equation}
    This is minimised by taking $|A| = m-r-1$, which reproduces the distance in~\cite{hastings2018distillation}. Moreover, if $3r<m$, then $\Delta$ satisfies the conditions of Proposition~\ref{prop:monomial_divisibility}, and $\mathcal{Q}(\Delta, \Gamma)$ has a transversal $T$ gate, also reproducing~\cite{hastings2018distillation}.
\end{example}

\section{Explicit Construction}\label{sec:explicit_construction}

\subsection{Weighting with One Heavy Variable}

We now consider a very particular choice of the above weights, and a deliberate choice of downsets.
\begin{lemma}\label{lem:weighted_params}
    Let $\boldsymbol{s}$ be a vector of length $m$, where
\begin{equation}
    \boldsymbol{s} = (\underbrace{1, 1, 1, \ldots, 1, 1}_{m-1}, h),
\end{equation}
for some positive integer $h$. In words, $\boldsymbol{s}$ has a $1$ in every entry, except the last, where it has entry $h$. We further consider positive integers $w,r$ for which $w,h \leq r$ and $2r \leq m-2+h$. With these choices, we consider the downsets 
\begin{align}
    \Delta &= \mathcal{W}_{\boldsymbol{s},r}\\
    \Gamma &= \{A \subseteq [m-1]: |A| \leq w\}.
\end{align}
In particular, $\Gamma$ is a subset of $2^{[m-1]}$, the set of subsets of $[m-1]$. Then, the corresponding quantum code arising from Proposition~\ref{prop:downsets_to_quantum_code} is well-defined, and has parameters $[[n,k,d]]$, where
\begin{align}
    n &= 2^m - k\\
    k &= \sum_{i=0}^{w}\begin{pmatrix}
        m-1\\i
    \end{pmatrix}\\
    d &= \min\left\{\begin{pmatrix}r+1\\>w\end{pmatrix}, 2^{r-h+1} + \begin{pmatrix}r-h+1\\>w\end{pmatrix}\right\},\label{eq:wrm_distance_minimum}
\end{align}
where $\left(\begin{smallmatrix}A \\ >B\end{smallmatrix}\right) \coloneq \sum_{i=B+1}^A \left(\begin{smallmatrix}A\\i\end{smallmatrix}\right)$, and the sum is taken to be zero if $B+1 > A$.

Moreover, if the three conditions
\begin{align}
    3r &< m-1+h,\label{eq:weighted_divis_condition_first}\\
    2r &< m-1,\\
    r &< m-2\label{eq:weighted_divis_condition_last}
\end{align}
all hold, then the corresponding quantum code supports a transversal $T$ gate.
\end{lemma}
\begin{proof}
    The downset $\Delta$ contains two types of points: those that are subsets of $[m-1]$, and have size at most $r$, and those that contain the final variable $m$, and at most $r-h$ out of the variables in $[m-1]$. Thus, the requirement that $w \leq r$ gives us $\Gamma \subseteq \Delta$. By Proposition~\ref{prop:dual_weighted_rm}, we have $\Delta^\perp = \mathcal{W}_{\boldsymbol{s},m-2+h-r}$, and so the requirement that $2r \leq m-2+h$ gives $r \leq m-2+h-r$, and therefore $\Delta \subseteq \Delta^\perp$. Proposition~\ref{prop:downsets_to_quantum_code} gives a well-defined quantum code $\mathcal{Q}(\Delta,\Gamma)$ with $|\Gamma| = \left(\begin{smallmatrix}m-1\\\leq w\end{smallmatrix}\right) \coloneq \sum_{i=0}^w\left(\begin{smallmatrix}m-1\\i\end{smallmatrix}\right)$ logical qubits.

    Let us now establish the distance of the quantum code via Equation~\eqref{eq:quantum_code_distance}. Consider any $A \in \Delta^\perp = \mathcal{W}_{\boldsymbol{s},m-2+h-r}$. First, suppose that $A$ does not contain the last variable of weight $h$, that is, $m \notin A$. Then, $A$ is some subset of $\{1, 2, \ldots, m-1\}$ of size $|A|\leq m-2+h-r$. We then have
    \begin{equation}
        \bar{A} = A' \sqcup \{m\},
    \end{equation}
    where $A'$ is some subset of $\{1, 2, \ldots, m-1\}$ of size $|A'| = (m-1)-|A| \geq r-h+1$. Consider a subset $V\subseteq \bar{A}$ that does not contain $m$. We have $V \in \bar{\Gamma}$ if and only if $V$ has size greater than $w$. There are $\left(\begin{smallmatrix}|A'|\\>w\end{smallmatrix}\right) \coloneq \sum_{i= w +1}^{|A'|}\left(\begin{smallmatrix}|A'|\\i\end{smallmatrix}\right)$ such subsets $V$, where we define the sum to be zero if $w + 1 > |A'|$. On the other hand, a subset $V \subseteq \bar{A}$ that does contain $m$ is automatically in $\bar{\Gamma}$. There are $2^{|A'|}$ such subsets $V$. In total, we have the number of subsets $V \subseteq \bar{A}$ for which $V \in \bar{\Gamma}$ as $2^{|A'|} + \left(\begin{smallmatrix}|A'|\\>w\end{smallmatrix}\right)$, whose attained minimum value is $2^{r-h+1} + \left(\begin{smallmatrix}r-h+1\\>w\end{smallmatrix}\right)$.

    On the other hand, suppose that $A \in \Delta^\perp = \mathcal{W}_{\boldsymbol{s}, m-2+h-r}$ does contain the last variable of weight $h$, that is, $m \in A$. Then, $A$ contains at most $m-2-r$ of the variables $\{1, \ldots, m-1\}$, as well as $m$. We then have that $\bar{A} \subseteq \{1, \ldots, m-1\}$, and $|\bar{A}| \geq r+1$. We then consider a subset $V \subseteq \bar{A}$. $V$ does not contain $m$, and so $V \in \bar{\Gamma}$ if and only if it has size greater than $w$. There are $\left(\begin{smallmatrix}|\bar{A}|\\>w\end{smallmatrix}\right)$ such subsets $V$. Here, the attained minimum value is $\left(\begin{smallmatrix}r+1\\>w\end{smallmatrix}\right)$, and we have the claimed distance.

    The transversality of the $T$ gate follows from Lemma~\ref{lem:divis_weighted_RM}, since one notes that Equations~\eqref{eq:weighted_divis_condition_first} to~\eqref{eq:weighted_divis_condition_last} give the conditions of that lemma, with the given choice of $\boldsymbol{s}$.
\end{proof}
\subsection{Explicit Asymptotic Parameters: Sub-Constant-Rate Regime}\label{subsec:explicit_asymptotic_sub_constant_rate}

We will now prove our results on the asymptotic parameters of explicit families of quantum CSS codes with transversal $T$ gates in the sub-constant-rate regime by developing the codes in Lemma~\ref{lem:weighted_params} into asymptotic families. Specifically, in this subsection, we will prove Theorem~\ref{thm:sub_constant_explicit}. We consider families of codes in Lemma~\ref{lem:weighted_params} indexed by the variable $m$. A given family will then be specified by choices of $w$ and $h$ for each $m$, that is, functions $w(m)$ and $h(m)$. One could choose to independently vary the value of $r$ for each $m$, that is, the choice of $r(m)$. However, we always let $r(m)$ be the largest integer such that Equations~\eqref{eq:weighted_divis_condition_first} to~\eqref{eq:weighted_divis_condition_last} are satisfied. Notice that if those equations are satisfied, then $2r \leq m-2+h$ is automatic.

We begin by considering a setup to achieve the best asymptotic parameters we can with this construction for which the rate is sub-constant. $w$ and $h$ are chosen to have a linear relation with $m$. That is, we let $w(m)$ and $h(m)$ be positive integers such that
\begin{align}
    w(m) &= p\cdot m + \mathcal{O}(1),\\
    h(m) &= q\cdot m + \mathcal{O}(1),
\end{align}
for constants $p \in (0,1/2)$ and $q \in [0,1]$. In this setup, a given family is specified exactly by our choices of the constants $p$ and $q$. We also write $r(m)$ as a sequence of positive integers, where
\begin{equation}
    r(m) = q'\cdot m + \mathcal{O}(1),
\end{equation}
for $q'= \min\left\{\frac{1}{2}, \frac{1+q}{3}\right\}$, which ensures Equations~\eqref{eq:weighted_divis_condition_first} to~\eqref{eq:weighted_divis_condition_last} (by taking the $\mathcal{O}(1)$ term in $r(m)$ small enough). We require $p,q < q'$ to ensure $w,h \leq r$. Given that, we see that we must have $q < \frac{1}{2}$, and therefore
\begin{equation}
    q' = \frac{1+q}{3}.
\end{equation}
Turning to the quantum code parameters, we have the length of the code
\begin{equation}
    n = 2^m - k,
\end{equation}
where the dimension is
\begin{equation}
    k \geq \begin{pmatrix}
        m-1\\w
    \end{pmatrix} = 2^{m(H_2(p) + o(1))},
\end{equation}
where $H_2(\cdot)$ is the binary entropy function, and we use the standard estimate for the asymptotic scaling of the binomial coefficient. Since $n \leq 2^m$, we obtain the dimension scaling
\begin{equation}
    k = \Omega\left(n^{H_2(p) + o(1)}\right).
\end{equation}
For the distance, we consider the two terms in the minimum in Equation~\eqref{eq:wrm_distance_minimum}. First considering $\left(\begin{smallmatrix}r+1\\>w\end{smallmatrix}\right)$, we note that if $p > \frac{q'}{2}$, the sum does not contain the centre of the binomial distribution, and so we have $\left(\begin{smallmatrix}r+1\\>w\end{smallmatrix}\right) = 2^{m\left(q'H_2(p/q')+o(1)\right)}$. On the other hand, if $p \leq \frac{q'}{2}$, the sum does contain the centre of the binomial distribution, and we have $\left(\begin{smallmatrix}r+1\\>w\end{smallmatrix}\right) = 2^{m\left(q'+o(1)\right)}$. It is convenient to define the function $\Psi$ given by
\begin{equation}
    \Psi(x) \coloneq \begin{cases}
        1 &\text{ if } x \in [0,1/2]\\
        H_2(x) &\text{ if } x \in [1/2,1]
    \end{cases}.
\end{equation}
We then have
\begin{equation}
    \begin{pmatrix}
        r+1\\>w
    \end{pmatrix} = 2^{m\left(q'\Psi(p/q')+o(1)\right)}.
\end{equation}
For the other term, we have
\begin{equation}
    2^{r-h+1} + \left(\begin{smallmatrix}r-h+1\\>w\end{smallmatrix}\right) \geq 2^{r-h+1} = 2^{m\left(q'-q + o(1)\right)}.
\end{equation}
The aim then becomes, for a fixed $p$ (specifying the scaling of the dimension) to pick the $q$ giving the optimum distance scaling, that is, to choose $q$ maximising
\begin{equation}
    \min\left\{q'-q, q'\Psi(p/q')\right\} = \min\left\{\frac{1-2q}{3},\frac{1+q}{3}\Psi\left(\frac{3p}{1+q}\right)\right\}
\end{equation}
subject to
\begin{equation}\label{eq:wrm_subconstant_feasibility}
    0 \leq q < \frac{1}{2}\;\;\;\text{ and }\;\;\; 3p-1 < q.
\end{equation}
We denote by $q^*(p)$ the optimum choice of $q$ given a value of $p$.
Let
\begin{align}
    A(q) &\coloneq \frac{1-2q}{3}\\
    B(q) &\coloneq \frac{1+q}{3}\Psi\left(\frac{3p}{1+q}\right).
\end{align}
We have that $A(q)$ is a strictly decreasing function of $q$, whereas $B(q)$ is strictly increasing in $q$, for any $p$, and so the optimum choice of $q$ is that with $A(q) = B(q)$, as long as that solution satisfies Equation~\eqref{eq:wrm_subconstant_feasibility}. We will show that this is
\begin{equation}
    q^*(p) = \begin{cases}
        0 &\text{ if } 0 < p \leq 1/6\\
        \text{The unique positive solution to } 1-2q = (1+q)H_2\left(\frac{3p}{1+q}\right) &\text{ if } 1/6< p < 1/2
    \end{cases}.
\end{equation}
Indeed, for $p \in (0,1/6]$, $A(0) = B(0)$, and $q=0$ satisfies Equation~\eqref{eq:wrm_subconstant_feasibility}. Now suppose that $p > 1/6$. If, in addition, $p < \frac{1}{3}$, then $q=0$ is still a feasible solution, i.e., it satisfies Equation~\eqref{eq:wrm_subconstant_feasibility}. We have $A(0) = \frac{1}{3} > B(0)$. On the other hand, if $p \geq \frac{1}{3}$, then $0$ is not a feasible solution, but we may check that as $q$ tends to $3p-1$ from above, we have $A(q) \to 1-2p > 0$, and $B(q) \to 0$. For all $p \in (\frac{1}{6},\frac{1}{2})$, as $q$ tends to $\frac{1}{2}$ from below, $A(q)$ tends to $0$, and $B(q)$ tends to a positive value. Thus, $A(q) = B(q)$ for a unique positive value of $q$. Moreover, it must be at a value for which $\Psi(x) = H_2(x)$, since otherwise one would have $\frac{1-2q}{3} = \frac{1+q}{3}$, giving $q = 0$.

We conclude that we may obtain parameters
\begin{align}
    k &= \Omega\left(n^{\alpha + o(1)}\right)\\
    d &= \Omega\left(n^{\beta_{\text{explicit}}(\alpha)+o(1)}\right),
\end{align}
where
\begin{equation}
    \beta_{\text{explicit}}(\alpha) = \begin{cases}
        \frac{1}{3} &\text{ if } 0 < \alpha \leq H_2(1/6)\\
        \frac{1-2q^*(\alpha)}{3} &\text{ if } H_2(1/6) < \alpha < 1
    \end{cases},
\end{equation}
where
\begin{equation}
    q^*(\alpha)\text{ is the unique positive solution to }1-2q = (1+q)H_2\left(\frac{3H_2^{-1}(\alpha)}{1+q}\right),
\end{equation}
where $H_2^{-1}:(0,1) \to (0,1/2)$ is the inverse binary entropy function. This concludes the proof of Theorem~\ref{thm:sub_constant_explicit}.
\begin{remark}
    Note that the parameters obtained with this weighted Reed-Muller construction coincide with the Hastings-Haah construction~\cite{hastings2018distillation} when $q=0$. In our construction, this is found to be the optimal choice in the regime $\alpha \leq H_2(1/6) \approx 0.650$, but not for $\alpha > H_2(1/6)$. Furthermore, we will shortly see that choosing $q>0$ will allow us to obtain codes with constant rate and growing distance.
\end{remark}

\subsection{Explicit Asymptotic Parameters: Constant-Rate Regime}\label{subsec:explicit_constant_rate}

In this subsection, we will prove Theorem~\ref{thm:constant_explicit} on the parameters we can achieve with explicit families in the constant-rate regime. This is essentially a more careful handling of the edge case $p,q \to \frac{1}{2}$ in the previous subsection. 

Once again, we will specify a family indexed by $m$ with functions $w(m)$ and $h(m)$, since $r(m)$ will always be chosen to be the maximum choice for which the transversality conditions, Equations~\eqref{eq:weighted_divis_condition_first} to~\eqref{eq:weighted_divis_condition_last}, are satisfied. As in the previous case, $2r \leq m-2+h$ will be satisfied automatically.

Let us suppose that $R \in (0,1/3)$ is the target constant rate we wish for our quantum code to achieve. Let $\kappa_R \coloneq \frac{2R}{1+R}$ and $c_R \coloneq -\Phi^{-1}(\kappa_R)>0$, where $\Phi$ is the cumulative distribution function of the standard Gaussian, noting that $c_R > 0$. We let $w(m)$ and $h(m)$ be positive integers satisfying
\begin{align}
    w(m) &= \frac{m}{2} - \frac{c_R}{2}\sqrt{m-1} + \mathcal{O}(1),\\
    h(m) &= \frac{m}{2} - 3c_R\sqrt{m-1}\left(\frac{1}{2}-\frac{2}{\log_2(m-1)}\right)+\mathcal{O}(1).
\end{align}
On the other hand, in order to ensure Equations~\eqref{eq:weighted_divis_condition_first} to~\eqref{eq:weighted_divis_condition_last}, we may let $r(m)$ be positive integers satisfying
\begin{equation}
    r(m) = \frac{m}{2}-c_R\sqrt{m-1}\left(\frac{1}{2}-\frac{2}{\log_2(m-1)}\right) + \mathcal{O}(1).
\end{equation}
One can check that this gives the conditions of Equations~\eqref{eq:weighted_divis_condition_first} to~\eqref{eq:weighted_divis_condition_last}, as long as the $\mathcal{O}(1)$ term in $r(m)$ is taken to be small enough. Moreover, one can see that $w,h \leq r$ are automatic.

Turning to the code parameters, we have $n = 2^m-k$, where the dimension $k$ is
\begin{equation}
    \frac{k}{2^m} = \frac{1}{2^m}\sum_{i=0}^w\begin{pmatrix}
        m-1\\i
    \end{pmatrix} = \frac{\kappa_R}{2} + o(1),
\end{equation}
using the central limit theorem for the binomial distribution.
We then have
\begin{equation}
    \frac{k}{n} = \frac{k}{2^m-k} = \frac{1}{\frac{2^m}{k}-1} = \frac{1}{\frac{2}{\kappa_R}-1} + o(1) = R + o(1).
\end{equation}
On the other hand, for the distance, we must again separately consider the two terms in the minimum of Equation~\eqref{eq:wrm_distance_minimum}. We begin with the first term. We have
\begin{equation}
    \begin{pmatrix}
        r+1\\>w
    \end{pmatrix} = \sum_{i=0}^{r-w}\begin{pmatrix}
        r+1\\i
    \end{pmatrix} \geq \begin{pmatrix}
        r+1\\r-w
    \end{pmatrix}.
\end{equation}
Therefore,
\begin{equation}
    \log_2\begin{pmatrix}
        r+1\\>w
    \end{pmatrix}\geq \log_2\begin{pmatrix}
        r+1\\r-w
    \end{pmatrix} \geq (r-w)\log_2\left(\frac{r+1}{r-w}\right)=\frac{2c_R\sqrt{m-1}}{\log_2(m-1)}\left(\frac{1}{2}\log_2(m)+o(\log_2(m))\right)
\end{equation}
and so
\begin{equation}
    \log_2\begin{pmatrix}
        r+1\\>w
    \end{pmatrix}\geq\sqrt{m}\left(c_R+o(1)\right).
\end{equation}
The second term in the minimum is $2^{r-h+1}+\left(\begin{smallmatrix}r-h+1\\>w\end{smallmatrix}\right)$, where we note the second term in this expression is zero for all $m$ large enough. On the other hand, we have
\begin{equation}
    r-h+1 = \sqrt{m}\left(c_R+o(1)\right),
\end{equation}
and so in total we obtain the distance scaling
\begin{equation}
    d \geq 2^{\left(c_R+o(1)\right)\sqrt{m}}.
\end{equation}
We have $\log_2(n) = m + O(1)$, and so $\sqrt{m} = \sqrt{\log_2(n)} + o(1)$, which gives the desired distance scaling, and the proof of Theorem~\ref{thm:constant_explicit}.

\section{Randomised Construction}\label{sec:random_construction}

In this section, we will develop our randomised construction. In the last section, we developed our explicit constructions by taking the downset $\Delta$ to come from a weighted Reed-Muller code, where the weighted Reed-Muller code uses the weights
\begin{equation}
    \boldsymbol{s} = (\underbrace{1, 1, 1, \ldots, 1, 1}_{m-1}, h),
\end{equation}
for some integer $h \geq 1$. That is, the first $m-1$ variables all have weight $1$, whereas the last variable can be heavier. In particular, we chose $\Delta = \mathcal{W}_{\boldsymbol{s},r}$ in the notation of Definition~\ref{def:weighted_RM}. We then punctured at some downset $\Gamma \subseteq 2^{[m-1]}$: in particular
\begin{equation}
    \Gamma = \{A \subseteq [m-1]: |A| \leq w\}
\end{equation}
for some $w \leq r$ so that $\Gamma \subseteq \Delta$. In this section, we will use the same form for the downset $\Delta$, but show that a more careful choice of the form of $\Gamma$ can yield improved parameters.

\subsection{Hypergraphs and Protected Points}

\begin{definition}[Hypergraph protecting the $y$-sets]\label{def:hypergraph_protecting_y_sets}
    Let $\mathcal{H}$ be a $t$-uniform hypergraph on $m-1$ points for some positive integer $t$. Let $y\geq t$ be a positive integer. We say that $\mathcal{H}$ \textit{protects the $y$-sets} if every subset of $[m-1]$ of size $y$ contains some edge of $\mathcal{H}$. We let $|\mathcal{H}|$ denote the number of hyperedges in a hypergraph $\mathcal{H}$.
\end{definition}
\begin{lemma}[Existence of a hypergraph protecting the $y$-sets]\label{lem:hypergraph_existence}
    Given positive integers $m,y,t$ such that $m-1 \geq y \geq t$, there exists a $t$-uniform hypergraph on $m-1$ points protecting the $y$-sets for which
    \begin{equation}
        |\mathcal{H}| \leq\left\lceil\frac{\begin{pmatrix}
            m-1\\t
        \end{pmatrix}}{\begin{pmatrix}
            y\\t
        \end{pmatrix}}\left(\ln\begin{pmatrix}
            m-1\\y
        \end{pmatrix}+1\right)\right\rceil.
    \end{equation}
\end{lemma}
\begin{proof}
    Suppose we define $\mathcal{H}$ by picking $L$ hyperedges, each of size $t$, by choosing each uniformly from the $\left(\begin{smallmatrix}m-1\\t\end{smallmatrix}\right)$ possibilities. Now, fix some $Y \subseteq [m-1]$ of size $y$. $Y$ contains $\left(\begin{smallmatrix}y\\t\end{smallmatrix}\right)$ sets of size $t$. The probability that a given random subset of size $t$ lies inside $Y$ is
    \begin{equation}
        p \coloneq \frac{\begin{pmatrix}
            y\\t
        \end{pmatrix}}{\begin{pmatrix}
            m-1\\t
        \end{pmatrix}}.
    \end{equation}
    Given $L$ independently chosen hyperedges, the probability that $Y$ contains none of them is
    \begin{equation}
        (1-p)^L \leq e^{-pL}.
    \end{equation}
    Union bounding over all possible $Y$, the probability that some $Y$ contains none of the $L$ hyperedges is at most
    \begin{equation}
        \begin{pmatrix}
            m-1\\y
        \end{pmatrix}e^{-pL}.
    \end{equation}
    This quantity is less than $1$ if
    \begin{equation}
        L = \left\lceil\frac{\begin{pmatrix}
            m-1\\t
        \end{pmatrix}}{\begin{pmatrix}
            y\\t
        \end{pmatrix}}\left(\ln\begin{pmatrix}
            m-1\\y
        \end{pmatrix}+1\right)\right\rceil.
    \end{equation}
    The size of $|\mathcal{H}|$ in the statement of the lemma is an upper bound rather than an equality since the random sampling can sample the same hyperedge multiple times.
\end{proof}
We let $E \in \mathcal{H}$ denote the situation where $E$ is a hyperedge of $\mathcal{H}$.
\begin{definition}[Punctured set with hypergraph protection]
    Given a non-negative integer $w$ and a $t$-uniform hypergraph $\mathcal{H}$ on $m-1$ points, define
    \begin{equation}
        \Gamma_{w, \mathcal{H}} = \{A \subseteq [m-1]: |A| \leq w \text{ and } E \not\subseteq A \text{ for all } E \in \mathcal{H}\}.
    \end{equation}
\end{definition}
We will think of $\mathcal{H}$ as protecting various points, like $y$-sets as in Definition~\ref{def:hypergraph_protecting_y_sets}. Then, $\Gamma_{w,\mathcal{H}}$ will be used as a puncture set, which does not puncture protected points. Note that $\Gamma_{w, \mathcal{H}}$ is a downset by definition.
\begin{definition}[Subcubes of the Boolean Hypercube]
    Given a point $L \subseteq [m-1]$, the subcube defined by $L$ is denoted $2^L \subseteq 2^{[m-1]}$, and is the set of subsets of $L$. Given points $E,L$ for which $E \subseteq L \subseteq [m-1]$, the cube defined by $E$ and $L$ is $[E,L] \coloneq \{V:E \subseteq V \subseteq L\}$.
\end{definition}
We can see that $2^L$ and $[E,L]$ may be viewed geometrically as subcubes of the Boolean hypercube. $2^L$ has size $2^{|L|}$, and $[E,L]$ has size $2^{|L|-|E|}$.
\begin{lemma}\label{lem:protected_puncture_dist_and_dim}
    Let $\mathcal{H}$ be a $t$-uniform hypergraph on $m-1$ points and $y \geq t$ be an integer. Suppose $\mathcal{H}$ protects the $y$-sets, and $|\mathcal{H}|$ is as in Lemma~\ref{lem:hypergraph_existence}. Let $w < y$ be a positive integer.
    \begin{enumerate}
        \item Suppose $L \subseteq [m-1]$ has $|L| \geq y$. Then
        \begin{equation}
            \left|2^L \setminus \Gamma_{w,\mathcal{H}}\right| \geq 2^{|L|-t}.
        \end{equation}
        \item We have
        \begin{equation}
            \left|\Gamma_{w,\mathcal{H}}\right| \geq \sum_{i=0}^{w}\begin{pmatrix}
                m-1\\i
            \end{pmatrix}\left[1-\left(\ln\begin{pmatrix}m-1\\y\end{pmatrix}+2\right)\frac{\begin{pmatrix}w\\t\end{pmatrix}}{\begin{pmatrix}y\\t\end{pmatrix}}\right].
        \end{equation}
    \end{enumerate}
\end{lemma}
\begin{proof}
    On the first point, $\mathcal{H}$ contains an edge $E \subseteq L$ for which $|E| = t$. We then have $[E,L] \subseteq 2^L \setminus \Gamma_{w,\mathcal{H}}$, and $[E,L]$ has size $2^{|L|-t}$.

    On the second point, we have that $\{A \subseteq [m-1]: |A| \leq w\}$ has size $\sum_{i=0}^w\left(\begin{smallmatrix}m-1\\i\end{smallmatrix}\right)$. One edge $E \in \mathcal{H}$ excludes $\left(\begin{smallmatrix}(m-1)-t\\i-t\end{smallmatrix}\right)$ points from membership in $\Gamma_{w,\mathcal{H}}$ of size $i$, for each $i = t, \ldots, w$. This is true because there are this many sets $A\subseteq [m-1]$ of size $i$ that contain $E$. One edge $E \in \mathcal{H}$ on its own thus excludes
    \begin{equation}
        \sum_{i=t}^w\begin{pmatrix}
            (m-1)-t\\i-t
        \end{pmatrix}
    \end{equation}
    points from membership in $\Gamma_{w,\mathcal{H}}$. Then,
    \begin{align}
        \sum_{i=t}^w\begin{pmatrix}
            (m-1)-t\\i-t
        \end{pmatrix}&\leq\frac{\begin{pmatrix}
            w\\t
        \end{pmatrix}}{\begin{pmatrix}
            m-1\\t
        \end{pmatrix}}\sum_{i=t}^w\begin{pmatrix}
            m-1\\i
        \end{pmatrix}\\
        &\leq \frac{\begin{pmatrix}
            w\\t
        \end{pmatrix}}{\begin{pmatrix}
            m-1\\t
        \end{pmatrix}}\sum_{i=0}^w\begin{pmatrix}
            m-1\\i
        \end{pmatrix}
    \end{align}
    Noting that $\sum_{i=0}^w\left(\begin{smallmatrix}m-1\\i\end{smallmatrix}\right)$ is the size of the puncture set if no points are protected, one edge $E \in \mathcal{H}$ on its own thus excludes a fraction at most $\frac{\left(\begin{smallmatrix}w\\t\end{smallmatrix}\right)}{\left(\begin{smallmatrix}m-1\\t\end{smallmatrix}\right)}$ from membership in $\Gamma_{w,\mathcal{H}}$. By a union bound, we thus have
    \begin{equation}
        |\Gamma_{w,\mathcal{H}}| \geq \sum_{i=0}^w\begin{pmatrix}
            m-1\\i
        \end{pmatrix}\left[1-|\mathcal{H}|\frac{\begin{pmatrix}
            w\\t
        \end{pmatrix}}{\begin{pmatrix}
            m-1\\t
        \end{pmatrix}}\right],
    \end{equation}
    and substituting in $|\mathcal{H}|$, we derive the claimed bound.
\end{proof}
In the next two subsections, we will construct quantum codes with transversal $T$ gates by choosing a downset $\Delta$ to define the original code, and then choosing the puncture set to be $\Gamma_{w,\mathcal{H}}$ for a given choice of $w$ and $\mathcal{H}$. Just as in the previous subsection, there will be a parameter range in which we want to consider $\Delta$ to be a downset associated with a typical Reed-Muller code, that is, $h=1$ in the weighted Reed-Muller construction, and there will be a parameter range in which we want to take a truly weighted Reed-Muller code, that is, $h>1$. Whereas these were considered in the same analysis in the explicit construction (see Subsection~\ref{subsec:explicit_asymptotic_sub_constant_rate}), we will consider them separately now. The two points in Lemma~\ref{lem:protected_puncture_dist_and_dim} will be used to establish the distance and dimension of the quantum codes in question. Indeed, since the number of logical qubits in our code construction (see Proposition~\ref{prop:downsets_to_quantum_code}) is equal to the size of $\Gamma_{w,\mathcal{H}}$, it is clear that the second point of Lemma~\ref{lem:protected_puncture_dist_and_dim} will help us there. On the other hand, referring back to Proposition~\ref{prop:downsets_to_quantum_code}, the expression for the distance of the quantum code may be succinctly rewritten
\begin{equation}
    d = \min_{L \notin \Delta}\left|2^L\setminus \Gamma\right|,
\end{equation}
and so the first point of Lemma~\ref{lem:protected_puncture_dist_and_dim} will be used to bound the distance. Before moving to the actual constructions, the following Lemma will be useful in bounding the size of $|\Gamma_{w,\mathcal{H}}|$, showing that it has essentially the same size as the unprotected set $\{A \subseteq [m-1]: |A| \leq w\}$, in a sensible asymptotic regime.
\begin{lemma}\label{lem:puncture_set_scaling}
    For positive integers $m$ tending to infinity, let $t(m), w(m), y(m)$ be sequences of positive integers satisfying
    \begin{equation}
        1 \leq t(m) \leq w(m) < y(m) < m-1.
    \end{equation}
    Let $\mathcal{H}(m)$ be a $t(m)$-uniform hypergraph on $m-1$ points protecting the $y(m)$-sets of size as in Lemma~\ref{lem:hypergraph_existence}. Suppose
    \begin{equation}\label{eq:puncture_size_condition}
        t(m)\ln\left(\frac{y(m)}{w(m)}\right) - \ln(m) \to \infty\text{ as } m \to \infty.
    \end{equation}
    Then,
    \begin{equation}
        |\Gamma_{w(m),\mathcal{H}(m)}| = (1-o(1))\sum_{i=0}^{w(m)}\begin{pmatrix}
            m-1\\i
        \end{pmatrix}.
    \end{equation}
    A sufficient condition for Equation~\eqref{eq:puncture_size_condition} is
    \begin{equation}\label{eq:puncture_size_sufficient}
        \frac{t(m)(y(m)-w(m))}{m} - \ln m \to \infty.
    \end{equation}
\end{lemma}
\begin{proof}
    We drop the $m$-dependence in this proof for brevity.
    
    $\Gamma_{w,\mathcal{H}}$ is contained in $\{A \subseteq [m-1]: |A| \leq w\}$ whose size is $\sum_{i=0}^w\left(\begin{smallmatrix}m-1\\i\end{smallmatrix}\right)$, and so it suffices to show that the term
    \begin{equation}
        \left(\ln\begin{pmatrix}
            m-1\\y
        \end{pmatrix}+2\right)\frac{\begin{pmatrix}
            w\\t
        \end{pmatrix}}{\begin{pmatrix}
            y\\t
        \end{pmatrix}}
    \end{equation}
    in Lemma~\ref{lem:protected_puncture_dist_and_dim} is $o(1)$. We have $\ln\left(\begin{smallmatrix}m-1\\y\end{smallmatrix}\right) + 2 \leq m + \mathcal{O}(1)$, and so it suffices to show that $\frac{\left(\begin{smallmatrix}w\\t\end{smallmatrix}\right)}{\left(\begin{smallmatrix}y\\t\end{smallmatrix}\right)}$ is $o(1/m)$. We have
    \begin{equation}
        \frac{\begin{pmatrix}
            w\\t
        \end{pmatrix}}{\begin{pmatrix}
            y\\t
        \end{pmatrix}} = \prod_{i=0}^{t-1}\frac{w-i}{y-i}\leq\left(\frac{w}{y}\right)^t.
    \end{equation}
    Equation~\eqref{eq:puncture_size_condition} says that the right-hand side is $o(1/m)$.

    Finally, we show that Equation~\eqref{eq:puncture_size_sufficient} implies~\eqref{eq:puncture_size_condition}. Indeed,
    \begin{equation}
        \ln\left(\frac{y}{w}\right) = -\ln\left(1-\frac{y-w}{y}\right) \geq \frac{y-w}{y} \geq \frac{y-w}{m},
    \end{equation}
    and therefore
    \begin{equation}
        t\ln\left(\frac{y}{w}\right) - \ln(m) \geq \frac{t(y-w)}{m} - \ln(m) \to \infty.
    \end{equation}
\end{proof}

\subsection{Puncturing the Reed-Muller Code with Protected Points}\label{subsec:rm_protect}

We again consider an asymptotic family of quantum codes indexed by $m$. In this subsection, we take the downset defining the un-punctured monomial code to be
\begin{align}
    \Delta(m) &\coloneq \{A \subseteq [m] : |A| \leq r(m)\}\\
    r(m) &\coloneq \left\lfloor \frac{m-1}{3}\right\rfloor,
\end{align}
that is, the un-punctured monomial code is the maximal $8$-divisible Reed-Muller code. Now fix some constant $p \in (0,1/3]$. Let $t(m), w(m), y(m)$ be sequences of positive integers satisfying
\begin{align}
    t(m) &= m^{1/2} + \mathcal{O}(1),\\
    w(m) &= p\cdot m - m^{3/4} + \mathcal{O}(1),\\
    y(m) &= r(m) + 1.
\end{align}
We let $\mathcal{H}(m)$ be a $t(m)$-uniform hypergraph protecting the $y(m)$-sets, with size as given in Lemma~\ref{lem:hypergraph_existence}; we will puncture the set $\Gamma_{w(m), \mathcal{H}(m)}$. Turning to Proposition~\ref{prop:downsets_to_quantum_code}, we note that we have $\Gamma \subseteq \Delta$, $\Delta \subseteq \Delta^\perp$, and $C(\Delta)$ is $8$-divisible. The code thus has $|\Gamma_{w(m),\mathcal{H}(m)}|$ logical qubits. Noting that Equation~\eqref{eq:puncture_size_sufficient} holds, and we have $1 \leq t(m) \leq w(m) < y(m) < m-1$, we have
\begin{equation}
    |\Gamma_{w(m), \mathcal{H}(m)}| = \Omega\left(2^{(H_2(p)+o(1))m}\right)
\end{equation}
by Lemma~\ref{lem:puncture_set_scaling}. As discussed in the previous subsection, an alternative expression for the distance of the quantum code from Proposition~\ref{prop:downsets_to_quantum_code} is
\begin{equation}
    d = \min_{L \notin\Delta}\left|2^L\setminus \Gamma\right|.
\end{equation}
Since $y(m) = r(m)+1$, we know that every $L \notin\Delta(m)$ has $|L|\geq y(m)$. If $L \subseteq [m-1]$, by the first point of Lemma~\ref{lem:protected_puncture_dist_and_dim}, we have
\begin{equation}
    \left|2^L\setminus\Gamma_{w(m), \mathcal{H}(m)}\right| \geq 2^{y(m)-t(m)}.
\end{equation}
On the other hand, if $m \in L$, then every element in $2^L$ containing $m$ is not in $\Gamma_{w(m), \mathcal{H}(m)}$, and there are at least $2^{y(m)-1}$ such subsets. We thus find that the distance of the quantum code is
\begin{equation}
    d \geq 2^{y(m)-t(m)} = 2^{m\left(\frac{1}{3}+o(1)\right)}.
\end{equation}
As always, $n \leq 2^m$, and so we obtain the asymptotic parameters of the quantum code
\begin{align}
    k &= \Omega\left(n^{H_2(p)+o(1)}\right),\\
    d &= \Omega\left(n^{\frac{1}{3}+o(1)}\right).
\end{align}
\subsection{Puncturing the Weighted Reed-Muller Code with Protected Points}\label{subsec:wrm_protect}

Whereas in the last subsection we punctured the ordinary Reed-Muller code with protected points to consider the dimension scaling $k = \Omega\left(n^{\alpha + o(1)}\right)$ for $\alpha \in (0, H_2(1/3)]$, here we will consider puncturing the weighted Reed-Muller code ($h > 1$) to consider $\alpha \in (H_2(1/3),1)$. To this end, we fix some constant $p \in (1/3,1/2)$, and let $h(m)$ and $r(m)$ be sequences of positive integers satisfying
\begin{align}
    r(m) &= p\cdot m + m^{3/4} + \mathcal{O}(1),\\
    h(m) &= (3p-1)\cdot m + 3\cdot m^{3/4} + \mathcal{O}(1).
\end{align}
As always, we consider $m$ weights
\begin{equation}
    \boldsymbol{s} = (\underbrace{1, 1, \ldots, 1}_{m-1}, h)
\end{equation}
and the corresponding choice of $\Delta$ is $\mathcal{W}_{\boldsymbol{s},r}$, as in Definition~\ref{def:weighted_RM} where we have dropped $m$-dependence here for clarity. We can see that $C(\Delta)$ is $8$-divisible from Lemma~\ref{lem:divis_weighted_RM}, as long as we take the $\mathcal{O}(1)$ term in $r(m)$ small enough (we also have $\Delta \subseteq \Delta^\perp$).

Next, we let $t(m)$ and $y(m)$ be sequences of positive integers satisfying
\begin{align}
    t(m) &= m^{1/2} + \mathcal{O}(1),\\
    y(m) &= r(m) + 1.
\end{align}
We then let $\mathcal{H}(m)$ be a $t(m)$-uniform hypergraph protecting the $y(m)$-sets, whose size is as in Lemma~\ref{lem:hypergraph_existence}. We also choose a sequence of positive integers $w(m)$ satisfying
\begin{equation}
    w(m) = p\cdot m + \mathcal{O}(1),
\end{equation}
and we will puncture the set $\Gamma_{w(m), \mathcal{H}(m)} \subseteq \Delta(m)$. Noting that $1 \leq t(m) \leq w(m) < y(m) < m-1$, we note that Equation~\eqref{eq:puncture_size_sufficient} holds, and using Lemma~\ref{lem:puncture_set_scaling}, we obtain
\begin{equation}
    k = \Omega\left(n^{H_2(p)+o(1)}\right),
\end{equation}
just as in the last subsection. Again, we use the expression for the distance $d = \min_{L \notin \Delta}\left|2^L\setminus \Gamma\right|$, dropping $m$-dependence for clarity. To compute this, we must consider two types of points $L$ not in $\Delta(m)$. First, if $L$ does not contain the last variable of weight $h(m)$, then it is a subset of $[m-1]$ of size at least $r+1$, and the first point of Lemma~\ref{lem:protected_puncture_dist_and_dim} gives us $\left|2^L\setminus \Gamma_{w(m), \mathcal{H}(m)}\right| \geq 2^{r(m)-t(m)+1}$. On the other hand, suppose that $L$ does contain the last variable of weight $h(m)$. Then its support on $[m-1]$ has size at least $r(m)-h(m)+1$, and since $\Gamma \subseteq 2^{[m-1]}$, we can guarantee that $\left|2^L\setminus \Gamma_{w(m), \mathcal{H}(m)}\right| \geq 2^{r(m)-h(m)+1}$ in this case. Since $h(m) > t(m)$, we obtain a distance lower bound $2^{r(m)-h(m)+1}$, i.e.,
\begin{equation}
    d = \Omega\left(2^{m\left(1-2p+o(1)\right)}\right).
\end{equation}
In conclusion, we obtain the distance scaling
\begin{equation}
    d = \Omega\left(n^{1-2H_2^{-1}(\alpha)+o(1)}\right),
\end{equation}
where $H_2^{-1}:(0,1) \to (0,1/2)$ is the inverse binary entropy function.

\subsection{Concatenating with Haah's Divisible Tower}

In subsections~\ref{subsec:rm_protect} and~\ref{subsec:wrm_protect}, we established the existence of asymptotic families of quantum codes with transversal $T$ gates and parameters
\begin{equation}
    \left[\left[n,\Omega\left(n^{\alpha + o(1)}\right),\Omega\left(n^{\beta_{\text{exist, WRM}}(\alpha)+o(1)}\right)\right]\right],
\end{equation}
where
\begin{equation}\label{eq:beta_exist_wrm}
    \beta_{\text{exist, WRM}}(\alpha) = \min\left\{\frac{1}{3}, 1-2H_2^{-1}(\alpha)\right\},
\end{equation}
and $H_2^{-1}:(0,1)\to(0,1/2)$ is the inverse binary entropy function. In this subsection, we will complete the proof of Theorem~\ref{thm:sub_constant_exist} by concatenating our construction with Haah's construction~\cite{haah2018towers}.

Haah demonstrates the existence of quantum CSS codes with transversal $T$ gates with parameters $[[n,\Omega(n^{1/2}), \Omega(n^{1/2})]]$. In that work, the notion of $T$ gate transversality is as follows: for each physical qubit $i \in [n]$, the gate $T^{a_i}$ is executed on qubit $i$ for some integer $a_i \in \{1,3,5,7\}$ and (without Clifford corrections), the logical $T$ gate is executed separately on each logical qubit. From this, one may construct a code with the same notion of $T$ gate transversality that we consider in this paper, without changing the asymptotic parameters up to constants. Indeed, we may ``repeat'' the $i$'th qubit $a_i$ times in each of its codestates. Slightly more formally, we may encode the $i$'th physical qubit of the code into the repetition code of length $a_i$ in the $Z$-basis. The resulting code has a transversal $T$ gate in the sense of this paper, that is, physical $T$ executed on each physical qubit separately leads to logical $T$ on each logical qubit separately, without Clifford corrections.

Now, if we have two quantum CSS codes with parameters $[[n_1, k_1, d_1]]$ and $[[n_2, k_2, d_2]]$, both with transversal $T$ gates, we may concatenate them to form quantum CSS codes with parameters $[[n_1n_2, k_1k_2, \geq d_1d_2]]$, with a transversal $T$ gate. Suppose that the first two codes come from asymptotic families where $k_i = n_i^{\alpha_i}$ and $d_i = n_i^{\beta_i}$ for $i = 1,2$ and for some constants $\alpha_i, \beta_i \in (0,1)$. Moreover, for some constant $\theta \in (0,1)$, suppose we choose $n_1 = N^\theta$ and $n_2 = N^{1-\theta}$. One checks that the resulting concatenated code has parameters $[[N,N^\alpha,N^\beta]]$, where
\begin{align}
    \alpha &= \theta \alpha_1 + (1-\theta)\alpha_2\\
    \beta &\geq \theta\beta_1 + (1-\theta)\beta_2.
\end{align}
That is, considering achievable points in the $(\alpha,\beta)$-plane, concatenation allows us to draw line segments between them to construct more achievable points. 

Haah's code gives an achievable point at $(1/2,1/2)$, and we may therefore achieve the exponents $(\alpha,1/2)$ for any $\alpha < 1/2$ by hard-coding logical qubits to fixed states to reduce the dimension. We may also consider the boundary of achievable exponent pairs in Equation~\eqref{eq:beta_exist_wrm} above. With this, and the achievable point $(1/2,1/2)$, the boundary of achievable exponent pairs that we may achieve is formed by joining the point $(1/2,1/2)$ to $(H_2(1/3), 1/3)$, giving the expression $\beta_{\text{exist}}(\alpha)$ in Theorem~\ref{thm:sub_constant_exist}.

\section*{Acknowledgements}

\addcontentsline{toc}{section}{Acknowledgements}

The author is grateful for the encouragement of Isaac Chuang in pursuing this project. The author acknowledges Anqi Gong for pointing them to the footprint bound, which led to the results in this work. 

The author acknowledges support from the MIT Department of Physics, from the MIT-IBM Watson AI Lab, and from NSF grant PHY-2325080. This preprint is assigned number MIT-CTP/6093.

\subsection*{Statement on AI Usage} All quantum code constructions were conceived by the author. The author acknowledges the use of ChatGPT 5.6 Sol for assistance in analysing the randomised construction in Section~\ref{sec:random_construction}, for generating code to make Figures~\ref{fig:sub_constant_achievable} and~\ref{fig:constant_achievable_parameters}, and for reviewing the paper. The author takes responsibility for all content.

\printbibliography

\end{document}